\documentclass{article}
\usepackage[novbox]{pdfsync}
\usepackage{amsfonts,amsmath,amssymb,mathrsfs,url,fancyhdr,
esint,comment,array,float}
\usepackage{todonotes}
\usepackage{booktabs}
\usepackage{enumitem}
\usepackage{geometry}
\usepackage{import,mathrsfs,float,pgfplots,mhchem,
bbm,graphicx,listings,microtype}
\usepackage[table]{xcolor} 
\usepackage{colortbl}
\usepackage{tcolorbox,multirow}\usepackage{etoolbox,booktabs}

\usepackage[toc,sort=use,nonumberlist]{glossaries}
\usepackage[utf8]{inputenc}   
\usepackage[T1]{fontenc}
\makeglossaries

\lstdefinestyle{mathematica}{
    language=Mathematica,
    basicstyle=\small\ttfamily,
    backgroundcolor=\color{gray!10},
    frame=single,
    breaklines=true,
    commentstyle=\color{green!60!black},
    keywordstyle=\color{blue},
    stringstyle=\color{red},
    showstringspaces=false,
    tabsize=2
}

\usepackage{nomencl}
\makenomenclature

\usepackage{makecell}
\usepackage[novbox]{pdfsync}
\usepackage{amssymb, amsfonts,mathtools,bbm,eurosym,graphicx,latexsym, amsmath,times,url}
\usepackage{graphicx,float,fancybox}
\usepackage{epsfig,epstopdf,tikz}
\usepackage[sans]{dsfont}
\usetikzlibrary{positioning,arrows,arrows.meta,calc}

\newcounter{tabcounter}
\newcounter{defcounter}
\newcounter{opcounter}
\newcounter{thmcounter}
\newcounter{excounter}

\newglossarystyle{unifiedstyle}{
  \setglossarystyle{list}
  
}

\usepackage[normalem]{ ulem }
\usepackage{soul}
\usepackage{cancel}
\usepackage{amsthm}
\usepackage{hyperref}
\newtheorem{open}{Problem}
\newtheorem{example}{Example}

\providecommand{\sep}{}
\providecommand{\bibcommenthead}{}

\usepackage{caption}
\usepackage{tikz}
\usepackage{pgfplots}
\pgfplotsset{compat=1.15}
\usepackage{mathrsfs}

\usetikzlibrary{
    shapes, shapes.geometric, shapes.symbols, shapes.callouts,
    arrows, arrows.meta,
    positioning, calc,
    decorations.pathmorphing, decorations.markings,
    backgrounds, fit,
    patterns, graphs, automata
}

\tikzset{
    startstop/.style={rectangle, rounded corners, minimum width=3cm,
        minimum height=1cm, text centered, draw=black, fill=blue!30},
    io/.style={trapezium, trapezium left angle=70, trapezium right angle=110,
        minimum width=3cm, minimum height=1cm, text centered,
        draw=black, fill=blue!30},
    process/.style={rectangle, minimum width=3cm, minimum height=1cm,
        text centered, draw=black, fill=green!30},
    arrow/.style={thick, ->, >=stealth},
    cloud/.style={draw, ellipse, fill=red!20, node distance=0.87cm,
        minimum height=2em},
    line/.style={draw, -latex'}
}
\usepackage{tikz-cd}
\tikzcdset{every label/.append style = {font = \small}}

\def\m0{{\mathcal R}_0}
\makeatletter
\newcommand{\beXa}[1][]{%
  \def\@temparg{#1}%
  \ifx\@temparg\@empty
    \begin{example}%
  \else
    \refstepcounter{excounter}
    \begin{example}[#1]%
    \label{e:#1}
    \newglossaryentry{e\theexcounter}{%
      name={#1},%
      description={p.~\pageref{e:#1}},
      user1={\theexcounter},%
      user2={Example}
    }%
  \fi
}
\newcommand{\eeXa}{\end{example}}

\newcommand{\beD}[1][]{%
  \def\@temparg{#1}%
  \ifx\@temparg\@empty
    \begin{definition}%
  \else
    \refstepcounter{defcounter}%
    \begin{definition}[#1]%
    \label{d:#1}%
    \newglossaryentry{d\thedefcounter}{%
      name={#1},%
      description={p.~\pageref{d:#1}},
      user1={\thedefcounter},%
      user2={Definition}%
    }%
  \fi
}
\newcommand{\eeD}{\end{definition}}

\newcommand{\beT}[1][]{%
  \def\@temparg{#1}%
  \ifx\@temparg\@empty
    \begin{theorem}%
  \else
    \refstepcounter{thmcounter}%
    \begin{theorem}[#1]%
    \label{t:#1}%
    \newglossaryentry{t\thethmcounter}{%
      name={#1},%
      description={p.~\pageref{t:#1}},%
      user1={\thethmcounter},%
      user2={Theorem}%
    }%
  \fi
}
\newcommand{\eeT}{\end{theorem}}
\newcommand{\beO}[1][]{%
  \def\@temparg{#1}%
  \ifx\@temparg\@empty
    \begin{open}%
  \else
    \refstepcounter{opcounter}%
    \begin{open}[#1]%
    \label{op:#1}%
    \newglossaryentry{o\theopcounter}{%
      name={#1},%
      description={p.~\pageref{op:#1}},
      user1={\theopcounter},%
      user2={Problem}%
    }%
  \fi
}
\newcommand{\eeO}{\end{open}}

\newcommand{\betA}[1][]{%
  \def\@temparg{#1}%
  \ifx\@temparg\@empty
    \begin{table}[htbp]%
  \else
    \refstepcounter{tabcounter}%
    \begin{table}[htbp]%

    \label{A:\thetabcounter}%

    \newglossaryentry{A\thetabcounter}{%
      name={#1},
      description={p.~\pageref{A:\thetabcounter}},%
      user1={\thetabcounter},%
      user2={Table}%
    }%
  \fi
}

\newcommand{\eetA}{\end{table}}
\makeatother

\newtheorem{lemma}{Lemma}
\newtheorem{proposition}{Proposition}
\newtheorem{corollary}{Corollary}

\newtheorem{remark}{Remark}
\newtheorem{question}{Question}
\newtheorem{assumption}{Assumption}

\newtheorem{definition}{Definition}
\newtheorem{theorem}{Theorem}

\def\beL{\begin{lemma}}\def\eeL{\end{lemma}}
\def\beP{\begin{proposition}}\def\eeP{\end{proposition}}
\def\beC{\begin{corollary}}\def\eeC{\end{corollary}}
\def\beR{\begin{remark}}\def\eeR{\end{remark}}
\def\beQ{\begin{question}}\def\eeQ{\end{question}}
\def\beA{\begin{assumption}}\def\eeA{\end{assumption}}

\definecolor{funccolor}{RGB}{25,25,112}
\definecolor{desccolor}{RGB}{64,64,64}

\def\bep{\begin{pmatrix}}\def\eep{\end{pmatrix}}
\def\bev{\begin{vmatrix}}\def\eev{\end{vmatrix}}
\def\bea{\begin{eqnarray*}}\def\eea{\end{eqnarray*}}
\def\bc{\begin{cases}}\def\ec{\end{cases}}
\def\BEN{\begin{enumerate}}\def\EEN{\end{enumerate}}
\def\BI{\begin{itemize}}\def\EI{\end{itemize}}

\newcommand{\be}[1]{\begin{equation}\label{#1}}
\newcommand{\ee}{\end{equation}}
\newcommand{\beq}{\begin{eqnarray}}
\def\eeq{\end{eqnarray}}
\def\eqr{\eqref}\def\fr{\frac}\def\lbl{\label}
\def\Lra{\Longrightarrow}\def\cNGM{\cite{Diek,Van,Van08}}
\def\Eq{\Leftrightarrow}

\newcommand{\R}{\mathbb{R}}  

\def\w{\omega}

\def\ga{\gamma}
\def\La{\Lambda}

\newcommand{\bff}[1]{{\mbox{\boldmath$#1$}}}

\def\v1{\vec {\bff 1}}

\long\def\symbolfootnote[#1]#2{%
\begingroup
\def\thefootnote{\fnsymbol{footnote}}\footnote[#1]{#2}%
\endgroup}

 \def\Gam{\Gamma^{-}}

\def\sd{s_{0}}

\def\ch{characteristic polynomial}

\def\eval{eigenvalue}

\def\ie{i.e. }
\def\im{\item}

\def\nne{non-negative}

\def\para{parameter}

\def\resp{respectively}
\def\sats{satisfies}
\def\ssec{\subsection}\def\sssec{\subsubsection}

\def\wk{well-known}

\newcommand\CRN{chemical reaction networks}

\def\MM{Michaelis-Menten}
\def\sm{stoichiometric matrix}

\def\brn{basic reproduction number}

\def\DFE{disease free equilibrium}

\def\ME{mathematical epidemiology}

\def\NGM{next generation matrix}

\def\bfp{boundary fixed point}

\def\RH{Routh-Hurwitz}

\def\reg{regular splitting}

 \def\rk{rich kinetics}
\def\cV{\cite{Vas,VasSta,BSV,GSV}}
\def\spf{semi-parametric -semi functional}
\def\CRW{Capasso-Ruan-Wang}
\begin{document}
\title{Triangularity of the Jacobian on siphon faces, the Metzler property of its transversal component, and other results}

\author{Florin Avram\thanks{Laboratoire de Math\'{e}matiques Appliqu\'{e}es, Universit\'{e} de Pau, 64000, Pau, France. \texttt{Florin.Avram@univ-Pau.fr}}
\and Rim Adenane\thanks{Laboratoire d'Analyse, G\'eom\'etrie et Applications, D\'epartement des Math\'ematiques, Universit\'e Ibn-Tofail, K\'enitra, 14000, Morocco. \texttt{rim.adenane@uit.ac.ma}}
\and Andrei-Dan Halanay\thanks{Department of Mathematics and Computer Science, University of Bucharest, Bucharest, RO-010014, Romania. \texttt{halanay@fmi.unibuc.ro}}}

\maketitle
\begin{abstract}
We establish two structural properties of the Jacobian of a positive ODE on any invariant boundary face, \ie\ on the face defined by a siphon of an associated \CRN.
First, the Jacobian is block lower-triangular at every point of such a face, with a transversal diagonal block, which governs the species that vanish on the face, and a tangential one.
Second, the transversal block is always a Metzler matrix.

These two facts yield an elegant proof, in the spirit of \CRN\ theory, and a generalization of the most cited result in \ME\ (ME), the \NGM\ (NGM) theorem, and they clarify its hypotheses. The Metzler property is exactly what makes the threshold $\rho(FV^{-1})<1$ work, it implies that \reg s of the transversal block always exist, and, via Perron--Frobenius theory, that a boundary fixed point can never lose its stability through a Hopf bifurcation in the directions transversal to its face.

Among other results, we review the ``symbolic-numeric" approach of Vassena and Stadler, which tackles bifurcation problems by viewing the \ch\  of the Jacobian at fixed points as a formal polynomial in the ``symbolic reactivities", and identifies its coefficients as ``Child Selection minors of the stoichiometric matrix". We also review two applications of this approach, to an SIRWS model and to a Capasso-Ruan-Wang family of SIRS models,
using the Mathematica package \texttt{Epid-CRN} at \url{https://github.com/florinav/EpidCRNmodels}, which implements tools from both \CRN\ and \ME.
\end{abstract}

\noindent\textbf{Keywords:}
biochemical interaction  networks; essentially nonnegative/positive systems; invariant faces; siphons/semi-locking sets; block-triangular Jacobian; Metzler matrix; Perron--Frobenius theory; transversal stability; mathematical epidemiology; disease free equilibrium;
regular splitting; next generation matrix; chemical reaction networks; stoichiometric matrix;
autocatalytic cores; symbolic-numeric Jacobian; symbolic-numeric Child Selection; rich kinetics;
Blokhuis-Stadler-Vassena oscillation recipes; SIR model; nonlinear force of infection;  bifurcation analysis; periodic solutions

 \tableofcontents

\section{Introduction}
Boundary faces of the positive orthant which are invariant under the flow of a positive ODE correspond to the siphons of an associated \CRN\ (CRN) \cite{AdLS}, and they carry the boundary fixed points of the system, like the disease free equilibria (DFE) of \ME\ (ME). Our main results are two structural properties of the Jacobian at any point of such a face (Theorem \ref{t:NGM} and Proposition \ref{p:Metzler}):
\BEN
\im the Jacobian is block lower-triangular, with a \emph{transversal} diagonal block $J^\perp$, which governs the species vanishing on the face, and a \emph{tangential} one $J_{\mathrm{tan}}$, which is the Jacobian of the dynamics restricted to the face;
\im the transversal block $J^\perp$ is a Metzler matrix.
\EEN
The first property is a known geometric fact for invariant coordinate subspaces, often observed at the DFE, but the second one seems not to have been stated explicitly, although it is precisely what makes the \NGM\ (NGM) approach of \cite{Diek,Van} work. As consequences, we obtain a CRN-flavored generalization of the NGM theorem to arbitrary invariant faces, with the correct hypotheses (Corollary \ref{c:NGM} and Remark \ref{r:hyp}), the automatic existence of \reg s (Corollary \ref{c:split}), which settles a question we raised in \cite{AABH25}, and Perron--Frobenius consequences (Corollary \ref{c:PF}). In particular, transversal stability can only be lost through a real eigenvalue, and never through a Hopf bifurcation.

This essay originated in a paper on the importance of reproducible scientific computation \cite{AAFG}, where we provided a  correction of  a famous paper (294 citations, and growing) concerning Hopf  and Bogdanov-Takens bifurcations due to L. Zhou and  M. Fan, ``Dynamics of an SIR epidemic model with limited medical
resources revisited",   in which we discovered and reported a significant numerical omission in  Figure \cite[Fig. 6]{ZhouFan}.


In the time since, our computational
efforts lead us to discover the parallel \CRN\ literature, with its numerous computing packages, to develop a Mathematica package \texttt{Epid-CRN} at \url{https://github.com/florinav/EpidCRNmodels}, which implements tools from both \CRN\ and \ME, to throw new light on the Zhou-Fan paper in \cite{VAA}, using the Vassena-Stadler ``Child Selections of the symbolic Jacobian" approach,  and finally to establish the structural results above. The ``other results" of this paper review these developments.

\paragraph{Contents}   We start in section \ref{s:pODE} by reviewing some fundamental concepts in the theory of positive systems, like \CRN\ representations, siphons and \reg s.

 We proceed in section \ref{s:NGM} with the triangularity of the Jacobian on invariant faces, the Metzler property of its transversal block, and their consequences: the generalized NGM theorem, the existence of \reg s, and the Perron--Frobenius properties of transversal stability.

 Section \ref{s:Ex} illustrates these results on three epidemic models of increasing generality: an SEIR model, whose transversal block is two-dimensional, the SIRWS model with immunity boosting, and the ``Capasso-Ruan-Wang" ``\spf" family of SIRS-type models with nonlinear forces of infection and treatment \eqr{Vys}, which generalizes Zhou-Fan.

 The other results are presented in sections \ref{s:CS} and \ref{s:CRW}.
 In section \ref{s:CS}, we introduce symbolic Jacobians and sketch some ideas of  the Vassena-Stadler stability and bifurcation analysis. Revolutionarily, their approach splits the stability and bifurcation problems in two parts.
 First the algebra, which examines the possible existence, for some parameters,  of positive or purely imaginary eigenvalues for the  symbolic Jacobian, and secondly, the ``realization of these parameters for given rates".
 Their work \cite{Vas,VasSta} showed that realization is always possible (and straightforward algorithmically) for ``parameter rich kinetics" (Michaelis-Menten, Hill, etc), but does not include models which have even one mass-action reaction (in which case example specific efforts are required).
 We apply these ideas to SIRWS, and, in section \ref{s:BR}, we discuss the idea of \cite{BorRost}  to
   prove symbolically the presence of Hopf bifurcations (for a mass-action model) by examining a reduced subnetwork and by applying ``inheritance rules" \cite{Banaji,BBH25}.  Possible starting points for  subnetworks which admit already oscillation may be found systematically by the  Vassena-Stadler approach.

 Section \ref{s:CRW} reviews results of \cite{VAA} for the Capasso-Ruan-Wang family, notably Theorems 2 (which applies essentially to any ME model) and 3, which shows that bifurcations at endemic points require the reproduction function to be larger than one there.

 The appendices contain an index of the definitions, theorems, examples and open problems (Appendix \ref{a:index}), background material and proofs on child selections (Appendix \ref{a:CS}), and the output of our package \texttt{Epid-CRN} for the examples (Appendix \ref{a:out}).


\section{A bird's eye view of positive ODEs}\label{s:pODE}
Population dynamics,  ecology,  mathematical epidemiology (ME), virology, rumour spreading, social networks, chemical reaction networks theory (CRNT), are a few of the biological interaction networks (BIN) subfields, which study all positive dynamical systems, and have similar preoccupations: the existence and multiplicity of equilibria, their local and global stability, the occurrence of bifurcations,   persistence, permanence, extinction, etc.

\beD[Positive / non-negative ODE]
An ODE is called \emph{positive} \cite{rantzer2015scalable} or \emph{non-negative} \cite{haddad2010}
if the non-negative orthant
\[
\mathbb{R}^n_{\ge 0} := \{ x \in \mathbb{R}^n : x_i \ge 0,\; i = 1,\dots,n \}
\]
is forward invariant under the flow.
\eeD

\beD[CRN stoichiometric representation]
\label{d:crn}
A \emph{CRN stoichiometric representation} of a(ny) positive ODE
is a triple $(\mathcal{S},\Gamma,\mathbf{r})$, such that
\be{Ga} {x}'=\Gamma\,\mathbf{r}(x),\quad x\in\R^n_{\ge 0}\ee
where $\mathcal{S}$ is the index set of the variables (called \emph{species}),
$\Gamma\in\R^{n\times n_R}$ is a constant \emph{stoichiometric} matrix whose columns are
identified with \emph{reactions}, and $\mathbf{r}=(r_1,\dots,r_{n_R})^T$ is a vector of \emph{rates}, such that for each reaction $\rho$ there is a set of \emph{reactants}
 $$\Gam_\rho\subseteq\mathcal{S},\quad \Gam_\rho\supseteq\{i \in \mathcal{S}:\Gamma_{i\rho}<0\}$$
 (which may also contain catalysts, like $i$ in $S+i\to 2i$), for which it holds that:
\BEN  \im
$r_\rho$ depends only on variables in $\Gam_\rho$; \im
$r_\rho:\R^n_{\ge 0}\to\R_{\ge 0}$ is locally Lipschitz, strictly increasing in each
$x_i$ with $i\in\Gam_\rho$;
\im
$r_\rho$  satisfies
\[
  x_i=0 \text{ for some } i\in\Gam_\rho \;\Longrightarrow\; r_\rho(x)=0.
\]
\EEN

\eeD
\beR These conditions ensure forward invariance of $\R^n_{\ge 0}$.
\eeR

\beR Such representations exist for any positive ODE, but are non-unique.\eeR

Especially relevant to the CRN-ME connection is the notion of {\bf siphons}. These  are in one-to-one correspondence with forward invariant boundary faces, cf. Angeli et al.~\cite{AdLS},
and in fact this yields their simplest definition. However, for completeness, we add
(and complete for our purpose) the definition given in the CRN literature:

\beD[siphon, minimal siphon, total/DFE siphon,  critical  siphon]\lbl{d:sif}
\BEN
\im A subset $S$ of species is a siphon/semi-locking set if every reaction that produces a species in S also consumes at least one species in $S$.

\im A siphon is said to be minimal if it does not properly contain any
other siphon.

\im The union of all minimal siphons will be called the total/DFE siphon.

\im A siphon is said to be critical if it does not contain the support of any
linear conservation.

\EEN
\eeD

 Siphons  are also related to the existence and stability of fixed points on the boundary of the positive orthant.
Indeed, Angeli et al.~\cite{AdLS} proved that if a boundary point is an $\omega$-$limit$ point, then the species set $\Sigma$ in which every species has zero concentration must be a siphon.

\beP[Boundary $\omega$-limits live in siphons]\lbl{p:AdLS} \cite[Prop 1]{AdLS}, \cite[Thm 2.5]{AndGAS}, \cite[Lem 2.1]{JohnSie}, \cite[Prop. 3.1]{FelShiu} If $w$ is an $\omega$-limit  of a strictly positive initial point for some choice of kinetics/rates, and belongs to a boundary face
$$w \in L_\Sigma:=\{X \in \mathbb{R}^{n}_{\geq 0}:x_i=0, i\in \Sigma\},$$
then $\Sigma$ is a siphon.
\eeP

Finally, we recall a concept that has been used an uncountable  number of times in ME, and which we believe relevant  also to CRNT.
\beD[regular splitting]
A matrix $J$ is said to admit   a regular splitting \cite{Varga} whenever $\exists (F,V)$ such that
$J=F-V$,  $V$ is invertible, and $F,V^{-1}$ are \nne.
\eeD

\section{Triangularity of the Jacobian on siphon faces and the Metzler property of its transversal block}\label{s:NGM}

Our main results come from the following circle of ideas:  \BEN \im For positive ODEs, the boundary faces defined by the zero coordinates of ``reasonable" boundary fixed points in the sense of \cite{AdLS} are invariant, and may be efficiently computed, by algorithms developed in the theory of Petri nets, under the name of siphons.
 \im  In mathematical epidemiology, \emph{epidemiologic strains} correspond conceptually to    \emph{minimal siphons}. \im The totality of ``infected variables" is the union of all minimal siphons,  the associated  face is the intersection of all invariant boundary faces, and usually it contains a unique  \DFE\ (DFE).
 \im The classic result of \cNGM\ may be decomposed in three parts: \BEN \im A lower triangular structure of the Jacobian on any siphon face, established below in Theorem \ref{t:NGM}, which reduces the stability problem to that of its diagonal (transversal and tangential) blocks.
 \im The Metzler property of the transversal block, established in Proposition \ref{p:Metzler}.
 \im The   existence for ME models of regular splittings $J^\perp=F-V$ \cite{Diek,Van}, which allow defining  next generating matrices, and, thanks to the Metzler property,  characterize the
  stability at the DFE (and at any fixed point on a siphon face, cf. Corollary \ref{c:NGM}),  by a (non-unique) inequality $\rho(F V^{-1})<1$.
   \EEN
  \im The positive eigenvalues of the NGM on  siphons, expressed as functions of the resident (non-siphon) variables, called reproduction functions in \cite{AAN,AAHJ,AABH25,AH},  characterize often existence and invasion domains of  \bfp s. Their sign relation with transversal stability at every point of the face is established in Remark \ref{r:repF}, but a precise statement concerning existence domains has not yet emerged; the observation can be made also for the \spf\ model of Sections \ref{s:CRWmodel} and \ref{s:CRW}.


 \EEN

Throughout this section, $\Sigma\subset\mathcal{S}$ is a set of species, $x\in\R^m_{\ge 0}$ collects the species in $\Sigma$, $y\in\R^n_{\ge 0}$ collects the remaining ones, and
\[
F_\Sigma:=\{(x,y)\in\R^{m+n}_{\ge 0}:\ x=0\}
\]
is the corresponding boundary face. We write the ODE as
\[
\dot x = f_x(x,y),\qquad \dot y = f_y(x,y),
\]
and assume that the vector field $f=(f_x,f_y)$ is positive, and $C^1$ on a neighbourhood of $\R^{m+n}_{\ge 0}$. At a point $(0,y)\in F_\Sigma$, we call
\[
J^\perp(y):=D_x f_x(0,y)\qquad\text{and}\qquad J_{\mathrm{tan}}(y):=D_y f_y(0,y)
\]
the \emph{transversal} and the \emph{tangential} components of the Jacobian. $J_{\mathrm{tan}}$ is the Jacobian of the ODE restricted to the face, while $J^\perp$ governs the linearized dynamics of the species of $\Sigma$, which vanish on the face.

\beT[Triangularity of the Jacobian on siphon faces]\label{t:NGM}
Assume that the face $F_\Sigma$ is {\bf invariant}, i.e.
\begin{equation}\label{eq:faceinv}
f_x(0,y)= 0\quad\text{for all }y\in\mathbb{R}^n_{\ge 0}
\end{equation}
(for a representation \eqr{Ga}, this holds whenever $\Sigma$ is a siphon, see Remark \ref{r:sifinv}).
Then:
\BEN
\im The mixed block $J_{xy}(y):=D_y f_x(0,y)$ vanishes at every point of the face, and hence the Jacobian is block lower-triangular there:
\[
J(0,y)=
\begin{pmatrix}
J^\perp(y) & 0\\
* & J_{\mathrm{tan}}(y)
\end{pmatrix},\qquad \forall\, y\in\R^n_{\ge 0}.
\]
In particular, the spectrum of $J(0,y)$ is the union of the spectra of $J^\perp(y)$ and $J_{\mathrm{tan}}(y)$.
\im A fixed point $(0,y^*)\in F_\Sigma$ is linearly asymptotically stable (\ie\ $J(0,y^*)$ is Hurwitz) iff both $J^\perp(y^*)$ and $J_{\mathrm{tan}}(y^*)$ are Hurwitz, and it is unstable as soon as one of them has an eigenvalue with positive real part.
\EEN
\eeT

\begin{proof}
By \eqr{eq:faceinv}, the map $y\mapsto f_x(0,y)$ vanishes identically on $\R^n_{\ge 0}$. Hence its derivative $D_y f_x(0,y)$ vanishes on the interior $\R^n_{>0}$, and, by continuity of $Df$, also at the boundary points of $\R^n_{\ge 0}$. Item 1 follows by writing
\[
J(0,y)=
\begin{pmatrix}
D_x f_x(0,y) & D_y f_x(0,y)\\
D_x f_y(0,y) & D_y f_y(0,y)
\end{pmatrix},
\]
and item 2 follows from item 1, since the characteristic polynomial of a block-triangular matrix is the product of those of its diagonal blocks.
\end{proof}

\beR\label{r:polyproof}
For polynomial (e.g.\ mass-action) systems, item 1 may also be read off the monomials: \eqr{eq:faceinv} forces every monomial of $(f_x)_i$ to contain some factor $x_p$, $p\in\Sigma$. Hence so does every monomial of $\partial (f_x)_i/\partial y_j$, which therefore vanishes at $x=0$.
\eeR

\beR\label{r:sifinv}
Assume $f=\Gamma\mathbf{r}$ is given by a representation \eqr{Ga}, let $i\in\Sigma$, and split
$(f_x)_i=\sum_\rho \Gamma^+_{i\rho}r_\rho-\sum_\rho\Gamma^-_{i\rho} r_\rho$, where $\Gamma^\pm_{i\rho}:=\max(\pm\Gamma_{i\rho},0)$. On $F_\Sigma$, the consumption terms vanish by property 3 of Definition \ref{d:crn}, since $i\in\Gam_\rho$ for each reaction $\rho$ with $\Gamma^-_{i\rho}>0$. If $\Sigma$ is a siphon, every reaction $\rho$ with $\Gamma^+_{i\rho}>0$ also has a reactant in $\Sigma$, so that the production terms vanish as well, and \eqr{eq:faceinv} holds. For the converse correspondence between invariant faces and siphons, see \cite{AdLS}.
\eeR

The second structural property of the Jacobian on invariant faces concerns the signs of the transversal block.

\beP[the transversal block is Metzler]\label{p:Metzler}
Under the assumptions of Theorem \ref{t:NGM}, the transversal block $J^\perp(y)$ is a Metzler matrix (\ie\ its off-diagonal entries are non-negative) at every point $(0,y)$ of the face $F_\Sigma$.
\eeP
\begin{proof}
Forward invariance of the orthant implies the boundary condition
\[
z\in\R^{m+n}_{\ge 0},\ z_i=0\ \Longrightarrow\ f_i(z)\ge 0
\]
(otherwise the solution started at $z$ would leave the orthant immediately), see \cite{haddad2010}. Let $i\ne j$ be two species of $\Sigma$, and let $t>0$. The point $z=(t e_j,y)$ is non-negative, with $z_i=0$, and hence $(f_x)_i(te_j,y)\ge 0=(f_x)_i(0,y)$, by \eqr{eq:faceinv}. Dividing by $t$ and letting $t\downarrow 0$ yields $J^\perp_{ij}(y)=\partial_{x_j}(f_x)_i(0,y)\ge 0$.
\end{proof}

\beR[a CRN proof, and the production--consumption splitting]\label{r:pcs}
Assume that $f=\Gamma\mathbf{r}$ is given by a representation \eqr{Ga} with $C^1$ rates, and write $\Gamma=\Gamma^+-\Gamma^-$ as in Remark \ref{r:sifinv}. Then $J^\perp(y)=F(y)-V(y)$, with
\[
F_{ij}(y):=\sum_\rho \Gamma^+_{i\rho}\,\partial_{x_j} r_\rho(0,y)\ge 0,\qquad
V_{ij}(y):=\sum_\rho \Gamma^-_{i\rho}\,\partial_{x_j} r_\rho(0,y), \qquad i,j\in\Sigma,
\]
where $F\ge 0$ since the rates are non-decreasing in each variable. Moreover, $V(y)$ is diagonal: if $\Gamma^-_{i\rho}>0$, then $i\in\Gam_\rho$, so $r_\rho$ vanishes on the hyperplane $\{x_i=0\}$, and so do its derivatives with respect to $x_j$, $j\ne i$. This yields a second proof of Proposition \ref{p:Metzler}, which identifies the off-diagonal entries of $J^\perp$ as rates of production of species of $\Sigma$ by other species of $\Sigma$.
\eeR

The Metzler property is what makes the classical next generation matrix argument work on siphon faces. Below, $s(A):=\max\{\Re \lambda:\lambda\in\sigma(A)\}$ denotes the spectral abscissa.

\beC[NGM theorem for siphon faces, generalizing \cite{Diek,Van,Van08,JA}]\label{c:NGM}
Let $(0,y)$ be a point of an invariant face $F_\Sigma$, and let $J^\perp(y)=F-V$ be a \reg\ (\ie\ $V$ is invertible and $F,V^{-1}$ are \nne). Then:
\BEN
\im $J^\perp(y)$ is Hurwitz iff $\rho(F V^{-1})<1$;
\im if moreover $V$ is a non-singular M-matrix (\ie\ $-V$ is also Metzler), then
$\mathrm{sign}\, s(J^\perp(y))=\mathrm{sign}\big(\rho(FV^{-1})-1\big)$.
\EEN
Consequently, if $(0,y^*)$ is a fixed point and $J_{\mathrm{tan}}(y^*)$ is Hurwitz, the fixed point is linearly asymptotically stable when $\rho(FV^{-1})<1$, and unstable when $\rho(FV^{-1})>1$ and $V$ is an M-matrix.
\eeC
\begin{proof}
Put $J:=J^\perp(y)$. Since $FV^{-1}=V(V^{-1}F)V^{-1}$, we have $\rho(FV^{-1})=\rho(V^{-1}F)$. The splitting $-J=V-F$ is regular in the sense of Varga, so that $\rho(V^{-1}F)<1$ iff $-J$ is invertible with $(-J)^{-1}\ge 0$ \cite[Ch.~3]{Varga}. By Proposition \ref{p:Metzler}, $-J$ is a Z-matrix, and a Z-matrix with a non-negative inverse is a non-singular M-matrix \cite[Ch.~6]{Berman}, \ie\ $J$ is Hurwitz. This proves item 1.

For item 2, note that the proof of item 1 only used the Metzler property of $J$. When $-V$ is Metzler, $J_t:=tF-V$ is Metzler for all $t\ge 0$, and $(tF,V)$ is a \reg\ of $J_t$. If $\rho(V^{-1}F)=1$, let $w\ge 0$, $w\ne 0$ be a Perron eigenvector of $V^{-1}F$. Then $Jw=Fw-Vw=0$, so $s(J)\ge 0$, while item 1 applied to $J_t$, $t<1$, gives $s(J_t)<0$, and hence $s(J)\le 0$ by continuity. If $s(J)\le 0$, then $J-\epsilon I=F-(V+\epsilon I)$ is Hurwitz for $\epsilon>0$, and $V+\epsilon I$ is still a non-singular M-matrix, so that item 1 gives $\rho(F(V+\epsilon I)^{-1})<1$, and letting $\epsilon\to 0$, $\rho(FV^{-1})\le 1$. Together with item 1, this shows that $s(J)<0$, $s(J)=0$, $s(J)>0$ iff $\rho(FV^{-1})<1$, $\rho(FV^{-1})=1$, $\rho(FV^{-1})>1$, \resp.

The last statement follows from Theorem \ref{t:NGM}.
\end{proof}

\beR[the hypotheses of Corollary \ref{c:NGM} are needed]\label{r:hyp}
Both statements fail for general matrices. The non-Metzler matrix $J=\begin{pmatrix}0&-1\\-1&0\end{pmatrix}$ admits the \reg\ $F=0$, $V=-J=V^{-1}$, with $\rho(FV^{-1})=0<1$, although $J$ has the eigenvalue $1$. The Metzler matrix $J=\begin{pmatrix}0&0\\1&0\end{pmatrix}$ admits the \reg\ $F=\begin{pmatrix}0&1\\2&0\end{pmatrix}$, $V=\begin{pmatrix}0&1\\1&0\end{pmatrix}=V^{-1}$, with $\rho(FV^{-1})=2>1$, although $s(J)=0$; here $V$ is not an M-matrix. By Proposition \ref{p:Metzler}, the first example can never arise as the transversal block of a positive ODE on an invariant face.
\eeR

\beR \label{rem:folk_theorem}
The block-triangular structure of the Jacobian on the DFE face can be read between the lines of classical mathematical epidemiology papers \cite{Diek, Van,Van08}; these authors probably viewed this fact as well-known and not worth mentioning. Indeed, from a classical dynamical systems perspective, the vanishing of the mixed Jacobian block is a known geometric consequence of evaluating a vector field's linearization on an invariant coordinate subspace (see, e.g., Hirsch, Smale, and Devaney \cite{HSD12}, or Wiggins \cite{Wig03}).

What is harder to intuit without a knowledge of CRNT, which deals explicitly with the boundary faces of the positive orthant, is that this structure is present at \textit{every} point of \textit{every} invariant face, in particular at any ``reasonable'' boundary equilibrium in the sense of Proposition \ref{p:AdLS}, and that it comes together with the Metzler property of the transversal block, on which the NGM threshold relies (Remark \ref{r:hyp}). This perspective allowed us to provide a natural set of assumptions for the classical NGM result.
\eeR

\beC[invariant faces always admit regular splittings]\label{c:split}
At every point $(0,y)$ of an invariant face, $J^\perp(y)$ admits a \reg\ with $V$ a non-singular M-matrix, for example
$V=cI$, $F=J^\perp(y)+cI$, for any $c>\max\big(0,\max_i(-J^\perp_{ii}(y))\big)$. Moreover, if all the diagonal entries $V_{ii}(y)$ of the production--consumption splitting of Remark \ref{r:pcs} are positive, then this splitting is regular, with $V(y)$ diagonal.
\eeC
\begin{proof}
Immediate from Proposition \ref{p:Metzler} and Remark \ref{r:pcs}.
\end{proof}

In \cite{AABH25}, we proposed to include the existence of a \reg\ on the DFE face among the defining conditions of ME models. Corollary \ref{c:split} shows that this condition holds automatically. However, the trivial splitting $V=cI$ yields $FV^{-1}=I+J^\perp/c$, whose spectral radius $1+s(J^\perp)/c$ has no epidemiological meaning, while different meaningful splittings may yield different reproduction numbers with the same threshold (see Example \ref{e:an SEIR model} in Section \ref{s:Ex}). This suggests the following question.

\beO[epidemiologically meaningful splittings on siphon faces]
Characterize, in terms of the reaction network, the regular splittings $J^\perp=F-V$ in which $F$ collects exactly the reactions creating new individuals of $\Sigma$ from individuals outside $\Sigma$ (``new infections''), and give conditions under which the corresponding $V$ is a non-singular M-matrix.
\eeO

A last group of consequences follows from the Perron--Frobenius theory of Metzler matrices.

\beC[Perron--Frobenius consequences]\label{c:PF}
Let $(0,y)$ be a point of an invariant face $F_\Sigma$. Then:
\BEN
\im The spectral abscissa $s(J^\perp(y))$ is a (real) eigenvalue of $J^\perp(y)$, with an eigenvector $v\ge 0$, $v\ne 0$, and all the other eigenvalues $\lambda\notin\R$ of $J^\perp(y)$ satisfy $\Re\lambda<s(J^\perp(y))$. If $J^\perp(y)$ is irreducible, then $v>0$ and $s(J^\perp(y))$ is simple.
\im At a fixed point $(0,y^*)$, a loss of transversal stability, \ie\ a change of sign of $s(J^\perp(y^*))$ along a path of parameters, happens through a real eigenvalue crossing $0$. In particular, if a stable fixed point on $F_\Sigma$ loses its stability through a Hopf bifurcation, then the crossing pair of complex eigenvalues belongs to $J_{\mathrm{tan}}(y^*)$.
\im If $s:=s(J^\perp(y^*))>0$ is not an eigenvalue of $J_{\mathrm{tan}}(y^*)$, then $J(0,y^*)$ has an eigenvector $(v,w)$ for the eigenvalue $s$, with $v\ge 0$, $v\neq 0$: the unstable direction points from the face into the positive orthant.
\EEN
\eeC
\begin{proof}
1. For $c$ large enough, $J^\perp(y)+cI\ge 0$, and the Perron--Frobenius theorem \cite{Berman} applied to this matrix yields the first and last claims. If $\lambda\notin\R$ is an eigenvalue with $\Re\lambda=s(J^\perp)$, then $|\lambda+c|>s(J^\perp)+c=\rho(J^\perp+cI)$, a contradiction.
2. At a crossing, $s(J^\perp)=0$, and by item 1 all the non-real eigenvalues of $J^\perp$ have negative real parts. If $J^\perp$ had a non-real eigenvalue on the imaginary axis, then $s(J^\perp)>0$, and the fixed point would not have been stable.
3. By the block-triangular structure, $(v,w)$ with $w:=(sI-J_{\mathrm{tan}})^{-1}(*)\,v$ is an eigenvector of $J(0,y^*)$.
\end{proof}

\beR[reproduction functions]\label{r:repF}
Since the matrix statements of Corollary \ref{c:NGM} hold at every point of the face, any family of regular splittings $J^\perp(y)=F(y)-V(y)$, with $V(y)$ a non-singular M-matrix, defines a \emph{reproduction function}
\[
\mathcal{R}_\Sigma(y):=\rho\big(F(y)V(y)^{-1}\big),\qquad (0,y)\in F_\Sigma,
\]
which satisfies $\mathrm{sign}\big(\mathcal{R}_\Sigma(y)-1\big)=\mathrm{sign}\, s(J^\perp(y))$ at \emph{every} point of the face, not only at fixed points. This gives a precise meaning to the reproduction functions of \cite{AAN,AAHJ,AABH25,AH}, and, together with item 3 of Corollary \ref{c:PF}, to the ``invasion'' observation in the last item of the list at the beginning of this section.
\eeR

\section{Examples of the main results}\label{s:Ex}

We illustrate the results of Section \ref{s:NGM} on three epidemic models of increasing generality. The SEIR model of Section \ref{s:SEIR} is the simplest one with a two-dimensional transversal block, the SIRWS model of Section \ref{s:SIRWS} will also serve as running example for the child selection methods of Section \ref{s:CS}, and the Capasso-Ruan-Wang family of Section \ref{s:CRWmodel} will be studied further in Section \ref{s:CRW}.

\ssec{An SEIR model}\label{s:SEIR}
\beXa[an SEIR model]
For the SEIR model
\[
\dot S=\Lambda-\beta S I-\mu S,\quad \dot E=\beta S I-(\sigma+\mu)E,\quad \dot I=\sigma E-(\gamma+\mu)I,\quad \dot R=\gamma I-\mu R,
\]
with reactions $0\to S$, $S+I\to E+I$, $E\to I$, $I\to R$ and four outflows, $\Sigma=\{E,I\}$ is the unique minimal siphon. Indeed, $\{E\}$ and $\{I\}$ are not siphons, since $S+I\to E+I$ and $E\to I$ produce $E$, \resp\ $I$, without consuming them. At a point $(S,0,0,R)$ of the face,
\[
J^\perp=\begin{pmatrix}-(\sigma+\mu)&\beta S\\ \sigma&-(\gamma+\mu)\end{pmatrix},
\]
which is Metzler, and irreducible when $S>0$. Its eigenvalues are real, and
\[
s(J^\perp)=-\frac{\sigma+\gamma+2\mu}2+\sqrt{\frac{(\sigma-\gamma)^2}4+\sigma\beta S}
\]
has a positive eigenvector.
The production--consumption splitting of Remark \ref{r:pcs} is
$F=\begin{pmatrix}0&\beta S\\ \sigma&0\end{pmatrix}$, $V=\mathrm{diag}(\sigma+\mu,\gamma+\mu)$, with $\rho(FV^{-1})=\sqrt{\mathcal{R}(S)}$, where
\[
\mathcal{R}(S):=\frac{\beta S\sigma}{(\sigma+\mu)(\gamma+\mu)}
\]
is the reproduction function yielded by the classical ``new infections'' splitting $F=\begin{pmatrix}0&\beta S\\ 0&0\end{pmatrix}$, $V=\begin{pmatrix}\sigma+\mu&0\\ -\sigma&\gamma+\mu\end{pmatrix}$ of \cite{Van}. Both splittings are regular, with $V$ an M-matrix, and yield the same threshold $s(J^\perp)>0\Leftrightarrow\mathcal{R}(S)>1$. At the DFE $S=\Lambda/\mu$, this is the classical condition $R_0=\frac{\beta\Lambda\sigma}{\mu(\sigma+\mu)(\gamma+\mu)}>1$.
\eeXa

\ssec{The SIRWS model with immunity boosting (4 species, 10 reactions, notebook \texttt{SIRWS.wl})}\lbl{s:SIRWS}

We introduce here an  epidemiological model which will serve us to illustrate the
utility of CRN methods. 

The ODE \begin{align}
\dot S &= \Lambda + \gamma_w W - \beta S i - \mu_s S \\
\dot i &= i(\beta S - \gamma_i - \mu_i) \\
\dot R &= \gamma_i i + \nu i W - (\mu_r + \gamma_r) R \\
\dot W &= \gamma_r R - \nu i W - \gamma_w W - \mu_w W
\end{align}
has species $\text{var} = \{S, i, R, W\}$ ($n=4$).

\beR We adopted the mixing of small letters and capitals  here and throughout since the latter, while being the traditional choice in ME, conflict with several reserved 
capitals in Mathematica, and hence create integration problems with our package. \eeR

\texttt{ODE2RN} decomposes this into $m = 10$ reactions:

\begin{center}
\begin{tabular}{clcl}
\hline
$j$ & Reaction & Rate $v_j$ & Description \\
\hline
1 & $0 \to S$           & $\Lambda$         & immigration \\
2 & $W \to S$           & $\gamma_w W$      & partial loss of  immunity \\
3 & $S + i \to 2i$      & $\beta S i$       & infection \\
4 & $i \to R$           & $\gamma_i i$      & recovery \\
5 & $i + W \to i + R$   & $\nu\, i\, W$     & boosting \\
6 & $R \to W$           & $\gamma_r R$      & total loss of  immunity \\
7 & $S \to 0$           & $\mu_s S$         & $S$ death \\
8 & $i \to 0$           & $\mu_i i$         & $i$ death \\
9 & $R \to 0$           & $\mu_r R$         & $R$ death \\
10 & $W \to 0$          & $\mu_w W$         & $W$ death \\
\hline
\end{tabular}
\end{center}

The unique minimal siphon is $\Sigma=\{i\}$.

The face of this siphon is $\{i=0\}$. At any of its points $(S,0,R,W)$, ordering the species as $(i\,|\,S,R,W)$, the Jacobian is block lower-triangular, as predicted by Theorem \ref{t:NGM}:
\[
J=\left(\begin{array}{c|ccc}
\beta S-\gamma_i-\mu_i & 0 & 0 & 0\\ \hline
-\beta S & -\mu_s & 0 & \gamma_w\\
\gamma_i+\nu W & 0 & -(\mu_r+\gamma_r) & 0\\
-\nu W & 0 & \gamma_r & -(\gamma_w+\mu_w)
\end{array}\right).
\]
The transversal block $J^\perp=\beta S-\gamma_i-\mu_i$ is scalar, hence trivially Metzler (see Example \ref{e:an SEIR model} for a two-dimensional transversal block). The production--consumption splitting of Remark \ref{r:pcs}, $F=\beta S$, $V=\gamma_i+\mu_i$, yields the reproduction function $\mathcal{R}(S)=\frac{\beta S}{\gamma_i+\mu_i}$. The tangential block has eigenvalues $-\mu_s,-(\mu_r+\gamma_r),-(\gamma_w+\mu_w)$, so that, by Corollary \ref{c:NGM}, the DFE $(S,i,R,W)=(\Lambda/\mu_s,0,0,0)$ is stable if
\[
R_0:=\mathcal{R}(\Lambda/\mu_s)=\frac{\beta\Lambda}{\mu_s(\gamma_i+\mu_i)}<1,
\]
and unstable if $R_0>1$. Note that the off-diagonal block $*$ has entries of both signs: Theorem \ref{t:NGM} and Proposition \ref{p:Metzler} give no sign information about it.

\ssec{The Capasso-Ruan-Wang family of SIRS models (notebook \texttt{SIRScrw.wl})}\label{s:CRWmodel}

 This work originated in studying stability and bifurcation results
for  \emph{Capasso-Ruan-Wang} nonlinear SIRS epidemic models, a class which includes many  models studied in the last decades -- see for example works by Capasso and Serio,
 Shigui Ruan, Wendy Wang, L. Zhou and  M. Fan, Vyska and Gilligan, E. Avila-Vales, Pei Yu and coauthors \cite{Capasso,WR,Wang,ZhouFan,Vyska,Rivero,Perez,LuRuan21,Pan21,Gupta}.

  We define this class via the ODE:
\begin{equation}\label{Vys}
\begin{cases}
s'= \Lambda      - \beta s iF[i] 
 -(\mu_s+\gamma_s) s + \gamma_r r+  i_s i,
\\
i'=
i(\beta s  F[i] -\gamma_i -\mu_i) -T[i], &\gamma_i=i_r+ i_s \\
r'= \gamma_s s + i_r   i +T[i] -(\gamma_r +\mu_r) r  ,
\end{cases}
\end{equation}
with $9$ independent {\bf positive parameters} $(\Lambda, \beta,\gamma_r,\gamma_s,i_s,i_s,\mu_s,\mu_i,\mu_r)$, and two functional ones:
a nonnegative, nondecreasing \emph{treatment rate} $T[i]\ge 0$, $T'(i)\ge 0, T(0)=0$, and a positive \emph{inhibition function} $F>0$ with $F(0)=1$.

 The occasional use of brackets for functional dependencies is convenient  for parsing our equations in Mathematica for such terms, but we abuse notation by using also parenthesis.

\beR The introduction of the direct return from infection to susceptibility
parameter $i_s$ is  used in SIR--type models for Gonorrhea, Influenza and Common cold, which share extremely weak and transient immunity, and frequent reinfections. Mathematically, this parameter may replace a
collapsing a short--lived recovered class.

\eeR

\beR We assume nonnegative $\Lambda,\gamma_s,\mu_s, \mu_i, \gamma_r, \mu_r \ge 0$ and strictly positive $\beta,\gamma_i>0$, since the existence of infection and recovery is essential to ME models, and has the surprising consequence that all models with these two interactions have symbolic Jacobians that ``admit purely imaginary \eval s under rich kinetics" -- see Thm \ref{lem:purelyimaginary}.
\eeR

For convenience, put $$
h(i):=(i_s+ i_r) i+T[i]=\gamma_i i+T[i] \quad\text{  with $\gamma_i>0$, $T[i]\ge 0$, $T'(i)\ge 0$.}\
$$

\paragraph{Historical account}
The SIRS model \eqref{Vys} generalizes several papers 
\cite{Capasso,LiuLevin,WR,Wang,ZhouFan,Vyska,LuRuan19,xu2021complex,Roostaei}. This epidemic model is firstly inspired by the pioneering work of Capasso and Serio \cite{Capasso}, which proposed to capture in $F(i)$ self-regulating adaptations of the population to the epidemics due to the spread of information on its existence. At the origin $F(i)$ must be $1$ to fit the Kermack-McKendrick SIR model, and the authors took  $F(i) =\frac{1}{1 + b i}$, where $b$ measures the  \emph{inhibition effect}. The ensuing infection rate $f[s,i]:=\beta s i F(i)$ then follows Michaelis--Menten in $i$.

A second strain of literature was initiated by Wang and Ruan \cite{WR,Wang} (see also
 \cite{Vyska}), who included a discrete treatment term $T[i]=\eta \operatorname{min}[i, w]$   which takes into account that some people recover only after having been treated in hospital. The minimum captures
 a possible breakdown of the medical system due to the epidemics. Subsequently, Zhang and Liu \cite{zhang2008backward} proposed to study  smooth treatment terms $T[i]=\eta \frac{i}{1+ \nu i}$.
In both versions
  $\eta=T'(0)$ is  the treatment rate when $i=0$,
and  $w=\frac{\eta}{\nu}$ is an upper limit for the capacity of treatment. Roostaei et al. \cite{Roostaei} further simplified the model by assuming $\eta=1$ and showed the occurrence of Hopf bifurcations. Putting together the nonlinear term of Capasso with that of Wang and Ruan gave rise to a huge literature focused on understanding the possible bifurcations which may arise. Further explorations concerned different nonmonotone
nonlinearities for $F(i)$  \cite{Ruan07, ZhouXiao, LuRuan19, Zhang23}.


It will be convenient to replace the infection rate $\beta s i F[i]$ by a general function $f[s,i]$, which yields the ODE
\begin{equation}\label{Vyss}
\begin{cases}
s'= \Lambda  -\gamma_s s    - f[s,i] 
+ \gamma_r r -\mu_s s+  i_s i,
\\
i'=
f[s,i] -i(\gamma_i +\mu_i) -T[i], \\
r'= \gamma_s s + i_r   i +T[i] -(\gamma_r +\mu_r) r
\end{cases}
\end{equation}

Both  ME and CRNT represent the model by a directed graph between the species, whose edges/reactions are   labeled by the rates.
 \begin{figure}[h!]
  \centering
 \begin{tikzpicture}[scale=.8,
      comp/.style={circle, draw, thick, minimum size=11mm, font=\Large},
      arr/.style={-stealth, thick},
      lbl/.style={font=\small, fill=white, inner sep=1pt}
    ]

    \node[comp] (S) at (0,0) {$S$};
    \node[comp] (I) at (3,3) {$I$};
    \node[comp] (R) at (6,0) {$R$};

    \draw[arr] (S) -- node[lbl, above left] {$\beta si F(i)$} (I);

    \draw[arr] (I) to[bend left=25] node[lbl, left=2pt, pos=0.5] {$i_s i$} (S);

    \draw[arr] (I) to[bend right=12] node[lbl, above=4pt, pos=0.3] {$i_r i$} (R);

    \draw[arr] (I) to[bend left=12] node[lbl, below=4pt, pos=0.5] {$T[i]$} (R);

    \draw[arr] (S) to[bend right=15] node[lbl, below] {$\gamma_s s$} (R);

    \draw[arr] (R) to[bend right=15] node[lbl, above] {$\gamma_r r$} (S);

    \draw[arr] (-1.8,0) -- node[lbl, above] {$\Lambda$} (S);

    \draw[arr] (S) -- node[lbl, left] {$\mu_s s$} +(0,-1.5);
    \draw[arr] (I) -- node[lbl, right] {$\mu_i i$} +(0,1.5);
    \draw[arr] (R) -- node[lbl, right] {$\mu_r r$} +(0,-1.5);

  \end{tikzpicture}
\caption{Flow diagram of model \eqr{Vys}. $\Lambda,\mu_s, \mu_i, \mu_r$ are the inflow rate of $S$ and the per
capita death rates of $S,I,R$, respectively.The linear conversion reactions
$
S   \underset{\gamma_s s}{\longrightarrow}   R,
R   \underset{\gamma_r r}{\longrightarrow}   S,    I   \underset{i_r i}{\longrightarrow}   R, I   \underset{i_s i}{\longrightarrow}   S$ model vaccinations, loss (waning) of immunity,   recovery, and very fast loss of immunity. The infection  and   {treatment rates}
$S+I   \underset{\beta s i F(i)}{\longrightarrow}   2I,
    I   \underset{T[i]}{\longrightarrow}   R$
are functional.}
  \label{fig:ZFG}
  \end{figure}

\beR

Note that the fact that the rate of the infection reaction depends on both $s$ and $i$ forces the CRN writing as $S+I   {\longrightarrow}   2I$, rather than
$S   {\longrightarrow}   I$, which is customary in ME flow diagrams, via an ``addition of I on both sides" (our package provides both versions, called respectively RN, and RNr, see Appendix \ref{a:out}). Note also that the RN  yields  a physical explanation of infection ``susceptible + infectious results in two infectious".
\eeR

\sssec{The DFE, its boundary stability analysis, and the reproduction function}

 Let $s_0$ indicate the value of $s$ at the DFE. The unique minimal siphon of the models \eqref{Vys}, \eqref{Vyss} is $\{i\}$, and by Theorem \ref{t:NGM}, the Jacobian at any point $(s,0,r)$ of the face $\{i=0\}$ is block lower-triangular, with scalar transversal block
\[
J^\perp(s)=f'_i(s,0)-\gamma_i-\mu_i-T'(0)=f'_i(s,0)-\mu_i-h'(0),
\]
where $f'_i(s,0)=\beta s F(0)=\beta s$ for \eqref{Vys}. The production--consumption splitting of Remark \ref{r:pcs} separates the production term $f'_i(s,0)$, due to the infection reaction $S+I\to 2I$, from the consumption term $\mu_i+\gamma_i+T'(0)$, due to the reactions consuming $I$, and it yields the reproduction function $\mathcal{R}(s)=\frac{f'_i(s,0)}{\mu_i+h'(0)}$ on the face (Remark \ref{r:repF}). Evaluating at the DFE, we obtain, for \eqref{Vys},
 \begin{equation}\label{R0}
   R_0 = \frac{\beta s_0}{\mu_i+\gamma_i + T'(0)},
\end{equation}

 Particular cases of this formula have appeared in the literature \cite{ZhouFan, Vyska, Gupta}; however,
$s_0$  is absent from \eqref{R0} in these papers, as they study a particular case with  $\gamma_s=0$, which implies $\sd=1$. Note that \eqref{R0} suggests that models in which $F(0)=0$ should perhaps not even be accepted as `epidemiological', since $R_0$ does not have any epidemiological interpretation. One  such example  from \cite{LuRuan19} is $ F(i)= \frac{i}{1+ b_1 i + b_2 i^2} \Lra F(0)=0$. Interestingly, this model shows a variety of dynamical behaviors as consequence of \emph{cusp} and \emph{Takens-Bogdanov bifurcations}. In later work\cite{Zhang23}, Zhang and Li showed that $F(i) =\frac{1+a i}{1  + b i^2}$ leads to similar features, with the significant improvement that $R_0$ is not identically  $0$.

\beR The fundamental formula \eqr{R0} suggests extending the reproduction function $\mathcal{R}(s)$ of the face $\{i=0\}$ to the whole orthant, by introducing
 \be{repF} R_0(S,i)=\frac{f'_i(S,i)}{\mu_i +h'(i)},\ee
 which coincides with $\mathcal{R}(S)$ for $i=0$.
 \eeR

\section{Other results I: symbolic Jacobians, child selections and inheritance}\label{s:CS}
The results of Section \ref{s:NGM} concern boundary fixed points. For bifurcations at positive (e.g.\ endemic) fixed points, no comparable structure is available, and one turns to the characteristic polynomial of the Jacobian. In this section, we review the symbolic Jacobian approach of Vassena and Stadler \cite{VasGB,VasSta,Vas,BSV,GSV}, and the inheritance theorems of Banaji and coauthors \cite{Banaji,GutShiu,BBH25}, and apply them to the SIRWS model of Section \ref{s:SIRWS}.

\beD[symbolic reactivity matrix, symbolic Jacobian]
Define the \emph{symbolic reactivity matrix}
$\mathcal{R}$ ($n_r \times n$), whose elements $r_{jk}$ are  formal positive indeterminates, replacing the partial derivative of rate $j$
with respect to species $k$, whenever this is non-zero.

The \emph{symbolic Jacobian} is the $n \times n$ matrix $G = \Gamma \cdot \mathcal{R}$, in which all the partial derivatives have been replaced by symbols.
\eeD

\beR The existence of any equilibria bifurcation  requires a choice of  symbols $r_{jk}$ so that the symbolic Jacobian $G$ is non-hyperbolic, i.e. possesses eigenvalues with zero real part.
\eeR

\beD[reactant and product matrices]
In CRNT,
 $\Gamma$ is further decomposed as
 \begin{equation}\label{Gam}
\Gamma_{kj}=\tilde{s}_k^j-s_k^j,
\end{equation}
where $s_k^j, \tilde{s}_k^j \geq 0$  called reactant and product coefficients, \resp, must satisfy that
$s_k^j>0$ {iff } the rate of reaction $j$ depends on species $k$.

 Then, $\Gamma=\tilde \alpha -\alpha,$  where both the reactant $\alpha$ and product matrices $\tilde \alpha$ contain only \nne\ elements.
\eeD

\paragraph{Complements to \RH\ analysis.}
Since the \eval s are rarely available in symbolic   ME models,  the \RH\ necessary and sufficient conditions for stability use instead the coefficients of the \ch. Still, the totality of the \RH\  conditions  leads  to complicated  computations which are typically beyond the power of nowadays computers, already in dimension 4 (especially due to the Hurwitz determinants). The limitations of \RH\ make   particularly attractive the  idea of {\bf ``Detailed-Routh-Hurwitz"} analysis (DRH), namely of extracting information from  partial summaries of the Jacobian matrix, like for example its
determinant.

This ties well to   Ivanova's idea~\cite{Ivanova}
of  establishing results like existence of Hopf bifurcations, chaos, etc first for smaller sub-networks, and then use  recipes which
guarantee inheritance to larger networks. Note this quest exists also in the general dynamical systems literature under the name of ``lifting bifurcations from invariant subspaces to the full system" \cite{Soares}.
This idea is  facilitated by using  a {\bf symbolic Jacobian} with functional increasing rates, and  decomposes its coefficients via the Cauchy-Binet formula (CB) into negative and positive summands (each summand being the determinant of a minor of the stoichiometric matrix); subsequently, special attention is paid to
 the  \emph{negative}  summands, which
``oppose the \RH\ criterion".

The relevant  minors, the ``Child Selections" \cite{Vas,VasSta} are square sub-matrices  of the \sm, rearranged to have  on the diagonal a reaction which produces the corresponding species, which translates into a negative sign for autocatalytic reactions (this reordering  relates them to the theory of Perron-Frobenius for Metzler matrices
    with nonnegative off-diagonal entries, which also governs the transversal block of the Jacobian on siphon faces, see Proposition \ref{p:Metzler} and Corollary \ref{c:PF}).

Conveniently, the signs of  the determinants/summands may be obtained via  symbolic decomposition of the   coefficients of the \ch\ of the {\bf symbolic Jacobian} into positive and negative terms \cite{Vas,VasSta,GSV}.
Amazingly, this approach yields existence of bifurcation results under any {\bf \rk\ \para s} (which include \MM, Hill, etc); however, for mass-action models,  further  work \cite{VasM}, or application of  inheritance extension criteria \cite{Banaji,GutShiu,BBH25,BorRost}, is required.

The idea of using inheritance ideas in ME was already illustrated in particular examples in \cite{VAA} and  Boros and Rost \cite{BorRost}. In this  section, we provide first a tutorial on how  possible network reductions may be obtained automatically  via the BiRNe method, and then illustrate on the SIRWS example how these may suggest
rigorous reductions  for mass-action models.

A historical overview of the ideas leading from critical fragments to child selections is given in Appendix \ref{a:CS}.

\beR
Inheritance results for reaction networks explain how oscillatory
mechanisms identified in small cores persist under biologically motivated model
extensions.
In this sense, inheritance theory complements the critical fragment / child
selection program by transferring local bifurcations from analytically tractable
cores to larger networks.
\eeR

\ssec{Summary of the Vassena-Stadler approach to  instability via   Cauchy--Binet expansions of symbolic Jacobians}\label{s:VS}

   \paragraph{CS-submatrices.}
   Given a CRN with $n$ species, and $n_r$ reactions,
let $\alpha=(s_i^j)$($n \times n_r$) and $\tilde\alpha=(\tilde s_i^j)$ ($n \times n_r$) denote the reactant and product matrices,
and let $G:=\tilde\alpha-\alpha$ ($n \times n_r$) denote the net stoichiometric matrix.

Let $ c_k$ denote the coefficients of the characteristic polynomial $\det(\Lambda I - G) = \Lambda^n + c_1 \Lambda^{n-1} + \cdots + c_n$.

By an elegant rewriting of the  Cauchy--Binet formula,  called Child-Selection (Theorem~5.5 in~\cite{Vas}),  it turns out that the coefficients $ c_k$ may be written in a form which makes the permutations appearing in the Cauchy Binet formula unnecessary for a CRN (or for a positive ODE).  This requires introducing:

\beD[Child selections as matchings in the reactivity matrix]
\label{def:CS-matching}
Let $\mathcal{N}$ be a reaction network with species set $M$ and reaction set $E$.
Let $\alpha=(s_i^j)_{i\in M,\,j\in E}$ be the reactant matrix and let
$R=(R_{j i})_{j\in E,\,i\in M}$ be the \emph{reactivity  matrix}, whose elements
are formal positive variable
$r_{j i}$ which
replace the rates derivatives in the symbolic Jacobian, and note that
\[
R_{j i}\neq 0 \quad\Longleftrightarrow\quad s_i^{j}>0.
\]

A \emph{$k$-child selection} (or $k$-CS) is a triple $\kappa=(\kappa,E_\kappa,J)$ such that
\[
\kappa\subseteq M,\quad E_\kappa\subseteq E,\quad |\kappa|=|E_\kappa|=k,
\]
and $J:\kappa\to E_\kappa$ is a bijection satisfying
\[
R_{J(i),i}\neq 0 \Eq s_i^{J(i)}>0
\qquad\forall\,i\in\kappa.
\]
We call $J$ the \emph{child-selection bijection} (matching species to reactions in which they
appear as reactants).
\end{definition}

Child -selections are implicit in the Cauchy–Binet formula.

\begin{lemma}[Cauchy--Binet monomials are exactly child selections]
\label{lem:CB-CS-exact}
Fix $\kappa\subseteq M$ with $|\kappa|=k$, and consider a $k\times k$ principal minor
$\det(G[\kappa])$.
\BEN
\im
For any $E_\kappa\subseteq E$ with $|E_\kappa|=k$, each term
$\det\!\bigl(R[E_\kappa,\kappa]\bigr)$ in the Cauchy--Binet expansion
\[
\det(G[\kappa])
=\sum_{E_\kappa\subseteq E,\ |E_\kappa|=k}
\det\!\bigl(\Gamma[\kappa,E_\kappa]\bigr)\;
\det\!\bigl(R[E_\kappa,\kappa]\bigr)
\]
contains monomials of the form
\[
\prod_{i\in\kappa} r_{J(i),i},
\]
where $J:\kappa\to E_\kappa$ is a bijection. Such a monomial is nonzero if and only if
$J$ is a child-selection bijection, i.e.\ if and only if $s_i^{J(i)}>0$ for all
$i\in\kappa$.

\im
Expanding $\det(R[E_\kappa,\kappa])$ by Leibniz yields
\[
\det\!\bigl(R[E_\kappa,\kappa]\bigr)
=
\sum_{J:\kappa\to E_\kappa\ \mathrm{bij.}}
\operatorname{sign}(J)\prod_{i\in\kappa} R_{J(i),i}.
\]
For each child selection $(\kappa,E_\kappa,J)$, let $P_J$ be the permutation matrix
associated with $J$, and set
\[
\Gamma[\kappa,J]:=\Gamma[\kappa,E_\kappa]\,P_J.
\]
Then
\[
\det(\Gamma[\kappa,J])
=
\operatorname{sign}(J)\det(\Gamma[\kappa,E_\kappa]).
\]
Hence the Cauchy--Binet expansion of $\det(G[\kappa])$ may be written as a sum over
child selections:
\[
\det(G[\kappa])
=
\sum_{(\kappa,E_\kappa,J)\in CS(\kappa)}
\det(\Gamma[\kappa,J])\prod_{i\in\kappa} R_{J(i),i},
\]
where $CS(\kappa)$ denotes the set of child selections supported on $\kappa$.
Thus the permutation signs are absorbed into the reordered stoichiometric determinants.
\EEN
\end{lemma}

The proof is given in Appendix \ref{a:CS}.

The signs of minors  are relevant for the stability of subnetworks due to the following \wk\ algebra  fact:
\beL[Determinant sign and instability mechanisms]
\label{lem:det-sign-instability}
Let $A\in\mathbb{R}^{k\times k}$ with $\det A\neq 0$, and define
\[
\sigma(A):=\operatorname{sign}\!\bigl((-1)^k\det A\bigr)\in\{-1,+1\}.
\]
Then:
\begin{enumerate}
\item If $\sigma(A)=-1$ (i.e.\ $(-1)^k\det A<0$), then $A$ is not Hurwitz stable.
Moreover, $A$ has at least one eigenvalue with $\Re\Lambda>0$.
\item If $\sigma(A)=+1$ (i.e.\ $(-1)^k\det A>0$), then the sign condition alone does not
decide stability. However, if $A$ is unstable, then instability cannot occur through an
\emph{odd} number of positive real eigenvalues. More precisely, at least one holds:
\begin{itemize}
\item[(i)] $A$ has a nonreal eigenvalue with $\Re\Lambda>0$ (hence a complex conjugate
pair in the open right half-plane);
\item[(ii)] $A$ has an \emph{even} number of real eigenvalues with $\Re\Lambda>0$ (counted
with algebraic multiplicity).
\end{itemize}
\end{enumerate}
\eeL

The proof is given in Appendix \ref{a:CS}.

\beR
 Within this framework, the sign of the sign-modified determinant
$\operatorname{sign}((-1)^k\det \Gamma[\kappa])$ separates CS-submatrices that
\emph{cannot be Hurwitz stable} (negative sign) from those for which Hurwitz stability is not
excluded by determinant sign alone (positive sign). The refinement ``exactly one unstable real
eigenvalue" versus ``exactly one unstable complex pair" requires an additional minimality
assumption (unstable core), stated in Lemma \ref{lem:UPF-NF-spectrum}. \eeR

\subsection{The UPF/NF dichotomy, unstable cores and spectral consequences}

In this section  we introduce terminology to classify child selections by determinant sign,
and examine the consequences of this classification on minimal unstable substructures.

\beD[Positive and negative feedbacks]
Let $\kappa$ be a $k$-child-selection and let $A := \Gamma[\kappa]$ be the associated
$k\times k$ CS-submatrix.
\begin{itemize}
\item $\kappa$ (or $A$) is called a \emph{negative feedback} (NF) if
\[
\operatorname{sign}\!\bigl((-1)^k \det A\bigr) = +1
\]
(see Lemma \ref{lem:UPF-NF-spectrum}).
\item $\kappa$ (or $A$) is called an \emph{unstable positive feedback} (UPF) if
\[
\operatorname{sign}\!\bigl((-1)^k \det A\bigr) = -1
\]
i.e. if the matrix $A$ is Hurwitz unstable (see Lemma \ref{lem:det-sign-instability}).
\end{itemize}
\eeD

\beR The UPF/NF dichotomy is purely algebraic and applies to arbitrary child
selections. The difference in  names and abbreviations between the two cases reflects the different spectral behaviors stated in Lemma \ref{lem:det-sign-instability}: an UPF cannot be Hurwitz stable, whereas a NF is not excluded from
being Hurwitz stable by determinant sign alone.
\eeR

\beD[Unstable core]
Let $\kappa$ be a $k$-child-selection and let $A := \Gamma[\kappa]$ be the associated
CS-submatrix.
We call $A$ an \emph{unstable core} if:
\begin{itemize}
\item $A$ is unstable (i.e.\ has at least one eigenvalue with $\Re\Lambda>0$), and
\item every proper principal submatrix of $A$ is Hurwitz stable.
\end{itemize}
\eeD

\medskip

Under the additional assumption that a UPF or NF arises from an unstable core,
the determinant-sign dichotomy acquires further spectral content.
\begin{lemma}[Instability mechanisms for unstable cores]
\label{lem:UPF-NF-spectrum}
Let $A$ be an unstable core of size $k$.

\begin{enumerate}
\item If $\operatorname{sign}((-1)^k\det A)=-1$ (UPF case), then $A$ has an odd number of
real eigenvalues with $\Re\Lambda>0$. In particular, instability may occur through a real
positive eigenvalue.

\item If $\operatorname{sign}((-1)^k\det A)=+1$ (NF case), then instability cannot occur
through an odd number of real positive eigenvalues. Since $A$ is an unstable core, this
forces instability to occur through exactly one complex conjugate pair with $\Re\Lambda>0$.
\end{enumerate}
\end{lemma}

The proof is given in Appendix \ref{a:CS}.

\ssec{Key results from~\cite{Vas} implemented in  BiRNe and EpidCRN}\label{s:VSkey}

The following key results from~\cite{Vas} form the basis of the BiRNe and EpidCRN implementations:

\begin{description}
\item[Corollary 5.12-5.14-5.11 (Instability-Stability).]

The existence of an UPF implies
that the network \emph{admits instability}
(there exist positive $r_{jk}$ making $G$ unstable).

In particular, autocatalytic networks admit instability.

\beR \emph{Stoichiometric autocatalysis}, as defined by Blokhuis et al. \cite{BLN}, has been proven equivalent to the presence of an unstable-positive feedback that is a Metzler matrix \cite[Thm. 7.3]{VasSta}. Note the analogy with Proposition \ref{p:Metzler}: on a siphon face, the transversal block of the Jacobian is always Metzler, and its instability is always detected by a real eigenvalue (Corollary \ref{c:PF}).
\eeR

If one $M$-CS ($M = |\mathbf{M}|$, the total number of species) has a Hurwitz-stable
stoichiometric submatrix $S[\mathbf{J}^{(M)}]$, then the network \emph{admits stability}.
Dually, if some $n$-CS $\mathbf{J}^{(n)}$ satisfies
\begin{equation}\label{badsign}
    \operatorname{sign}(-1)^{n}\alpha_{\mathbf{J}^{(n)}}=-1,
\end{equation}
then the network admits instability.

\item[Corollary 5.7, 5.8 (Fixed sign).]
A coefficient $c_k$ is of \emph{fixed sign}  for all positive $\mathbf{r}'$
if and only if the {sign modified determinants} of order $k$
 share the same sign, \ie\ for any two $k$-child-selections $\mathbf{J}_1^{(k)}, \mathbf{J}_2^{(k)}$ we have
$\alpha_{\mathbf{J}_1^{(k)}} \alpha_{\mathbf{J}_2^{(k)}} \geq 0$.

 The Jacobian $G$ is a $P^{(-)}_0$ matrix for all positive $\mathbf{r}'$ if and only if for all Child-Selections $\mathbf{J}^{(n)}$ it holds
$$\operatorname{sign}(-1)^{n}\alpha_{\mathbf{J}^{(n)}}=1.$$

\item[Theorem 5.19 (Purely imaginary eigenvalues).]
Fix $\bar{n}$ and a species set $\bar{\mathbf{M}}^{(\bar{n})}$.
Let $\mathcal{C}$ be a collection of $\bar{n}$-CS such that either
$\operatorname{sign} (-1)^{\bar{n}} \alpha_{\mathbf{J}^{(\bar{n})}} = 1$ or
$\alpha_{\mathbf{J}^{(\bar{n})}} = 0$ for all $\mathbf{J}^{(\bar{n})} \in \mathcal{C}$.
Assume:
\begin{enumerate}
\item There exists $\mathbf{J}_1^{(\bar{n})} \in \mathcal{C}$ such that
$S[\mathbf{J}_1^{(\bar{n})}]$ is a stable matrix.
\item There exists an $n$-CS $\mathbf{J}_2^{(n)}$ ($n < \bar{n}$) with
$\mathbf{M}^{(n)} \subset \bar{\mathbf{M}}^{(\bar{n})}$,
$(m, \mathbf{J}_2^{(n)}(i)) \in \mathcal{P}$ for all $m \in \mathbf{M}^{(n)}$,
and $S[\mathbf{J}_2^{(n)}]$ is unstable.
\end{enumerate}
Then there exists a choice of symbols $\mathbf{r}'$ such that $G(\mathbf{r}')$ has
purely imaginary eigenvalues.

\end{description}

\medskip
\beD [oscillatory core recipes of Blokhuis-Stadler-Vassena]
\noindent
By \cite[Recipe I,II,0]{BSV}, RNs with oscillatory core of class I,II,0 admit non-stationary periodic orbits.
\BEN \im
Recipe I: An \textbf{Oscillatory Core of Class I}
 is a pair (UPF, stable super-CS) where the super CS is minimal with the property of being Hurwitz-stable and possessing an unstable-positive feedback as a principal submatrix.

 \im Recipe II: An \textbf{Oscillatory Core of Class II} is a minimal $k\times k$ CS-matrix which is a unstable and negative feedback, and  which has a $(k-1)\times (k-1)$ Hurwitz-stable principal sub-matrix.
 \im \textbf{Recipe 0}:
  A transition between regions of parameter space with stable and unstable
  steady states, under the assumption that the Jacobian remains invertible
  throughout, suggests bifurcation. For example, cf.
   \cite[Corollary 5.21]{Vas} (Fixed-sign determinant),
if $G$ has a nonzero determinant of fixed sign, and
if the network admits both stability and instability, then there exists $\mathbf{r}'$
such that $G(\mathbf{r}')$ has purely imaginary eigenvalues.
\EEN
\eeD

\beR
In the remarkable case in which the Jacobian $G$ is of fixed sign \cite{Hirsch,Smith,AngeliSontag2003,CF05,BanCra,Feinberg2019} (see also Section 7 in the ArXiv version),
Vassena's Corollary~5.21 provides
 sufficient conditions for Hopf
that do not require set inclusion between UPF and stable CS. It suffices:
\begin{itemize}
\item The network admits instability (e.g.\ via autocatalysis, Corollary~5.14).
\item The network admits stability (for example, by Corollary~5.11: some $M$-CS matrix is stable).
\end{itemize}

\eeR

  The \cV\ approach has been implemented in Python in the BiRNe package \cite{GSV}, under the assumption that the \sm\ is real, and ported now in our package, which allows symbolic \sm s.  Note that admissibility (of Hopf bifurcations, for example), implies existence under \rk\ \para s (\MM, Hill, etc), but further example specific work is required for semi-parametric cases \cite{VAA}, or  parametric cases like mass-action models \cite{VasM}; alternatively, one may apply for mass-action inheritance extension criteria \cite{Banaji,GutShiu,BBH25,BorRost}.    We hope to  illustrate in further work that these approaches can  be  useful in \ME.

\ssec{Child selections and oscillatory cores of SIRWS}\label{s:VSsirws}
We now apply the results of Section \ref{s:VSkey} to the SIRWS model of Section \ref{s:SIRWS}, using \texttt{Epid-CRN}.

Characteristic polynomial coefficients $c_1, \ldots, c_4$ have
$\{10, 35, 51, 25\}$ monomials respectively, of which $\{1, 6, 10, 5\}$ are negative.

\BEN \im
There are 22 negative child selections ($1+6+10+5$ by order $k=1,2,3,4$);
13 are unstable positive feedbacks (UPFs) and 9 are negative feedbacks (NFs),
all UPFs containing the autocatalytic infection reaction $S+i\to 2i$.

\noindent\emph{Example UPF} (2-CS, species $\{i,R\}$, reactions $\{\beta Si,\,\mu_r R\}$):
\[
\Gamma[\{i,R\},\{3,9\}]
= \begin{pmatrix} 1 & 0 \\ 0 & -1 \end{pmatrix},
\qquad \det = -1,\quad (-1)^{2}\det = -1 < 0 \;\Rightarrow\; \text{UPF}.
\]
Eigenvalues $+1$ and $-1$: one real positive eigenvalue, confirming instability
(Lemma~\ref{lem:det-sign-instability}(1)).

\noindent\emph{Example NF} (2-CS, species $\{S,i\}$, reactions $\{\beta Si,\,\mu_s S\}$):
\[
\Gamma[\{S,i\},\{3,7\}]
= \begin{pmatrix} -1 & -1 \\ 1 & 0 \end{pmatrix},
\qquad \det = +1,\quad (-1)^{2}\det = +1 > 0 \;\Rightarrow\; \text{NF}.
\]
Eigenvalues $-\tfrac{1}{2}\pm\tfrac{\sqrt{3}}{2}\,i$: Hurwitz stable with a complex
conjugate pair, consistent with Lemma~\ref{lem:det-sign-instability}(2)
(stability is not excluded by the determinant sign, and here stability holds).
This NF 2-CS is precisely the minimal Hurwitz-stable superset of the UPF $\{i\}$
that constitutes the Recipe~I oscillatory core.

\im The
Hasse diagram has 71 inclusion edges before transitive reduction.
\EEN

  \begin{center}
  \begin{tabular}{clll}
  \hline
  \# & Species & Rates & Stable supersets \\
  \hline
  1 & $\{i\}$       & $\{\beta S i\}$                          & $\{S,i\},\; \{S,i,R,W\}$ \\
  2 & $\{i,R\}$     & $\{\beta S i,\; \mu_r R\}$               & $\{S,i,R,W\}$ \\
  3 & $\{i,W\}$     & $\{\beta S i,\; \mu_w W\}$               & $\{S,i,R,W\}$ \\
  4 & $\{i,R,W\}$   & $\{\beta S i,\; \mu_r R,\; \mu_w W\}$   & $\{S,i,R,W\}$ \\
  \hline
  \end{tabular}
  \captionof{table}{Four  pairs (UPF, minimal stable superset) found. All  share the infection autocatalysis reaction $S + i \to 2i$.}
  \label{t:UPF}
\end{center}

 \paragraph{Detailed analysis of pair \#4.}
  The stoichiometric submatrix of the UPF is:
  \[
  \Gamma[\{i,R,W\},\{3,9,10\}] = \operatorname{diag}(1,-1,-1),
  \qquad \det = +1 \;\Rightarrow\; \text{unstable positive feedback (UPF).}
  \]

  The unique Hurwitz-stable superset adds species $S$ matched to reaction $\mu_s S$ (death of susceptibles):
  \[
  \Gamma[\{S,i,R,W\},\{3,7,9,10\}]
  = \begin{pmatrix}
  -1 & -1 &  0 &  0 \\
   1 &  0 &  0 &  0 \\
   0 &  0 & -1 &  0 \\
   0 &  0 &  0 & -1
  \end{pmatrix}.
  \]
  This matrix is block-diagonal. The $\{S,i\}$ block
  $\bigl(\begin{smallmatrix} -1 & -1 \\ 1 & 0 \end{smallmatrix}\bigr)$
  has eigenvalues $-\tfrac{1}{2} \pm \tfrac{\sqrt{3}}{2}\,i$ (an oscillatory infection cycle),
  while the $\{R,W\}$ block is $-I_2$ (pure decay).
  However, these two blocks are not Oscillatory Core of Class~I, since they do not
  have an embedded UPF.

  Finally, $\Gamma[\{S,i,R,W\},\{3,7,9,10\}]$  is a minimal Hurwitz-stable superset containing an UPF as a principal submatrix,
   confirming again
    the structural capacity for Hopf bifurcation of the SIRWS model.

  The \texttt{oscRI} script (Recipe~I oscillatory cores) reports one minimal
  oscillatory core of Class~I (see Appendix \ref{a:out} for the output).
  The Hurwitz-stable stoichiometric minor (columns in sorted reaction order
  $\{\beta Si,\,\mu_s S\}$) is $\Gamma[\{S,i\},\{3,7\}]=\bigl(\begin{smallmatrix}-1&-1\\1&0\end{smallmatrix}\bigr)$,
  with eigenvalues $-\tfrac{1}{2}\pm\tfrac{\sqrt{3}}{2}\,i$ (an oscillatory infection cycle;
  the same block appears in the $4\times 4$ analysis above).
  Together with the UPF $\{i\}$, this is an Oscillatory Core of Class~I, confirming that
  SIRWS admits Hopf.

\paragraph{Recipe~II oscillatory cores and fixed-sign analysis.}
  Running \texttt{oscRII[nCS, gH, spe, rts]} and \texttt{fixSgn[nCS, n]} on SIRWS
  (see Appendix \ref{a:out} for the output) finds no oscillatory cores of Class~II, and shows that only the coefficient $c_1$ of the characteristic polynomial has fixed sign.
  The analytically certain entries are: $k=1$ has fixed sign (one UPF, $\det=1$);
  $k=2$ is not fixed sign (2 UPFs with $\det<0$ and 4 NFs with $\det>0$);
  $k=4$ is not fixed sign because the Hurwitz-stable $4\times 4$ NF child selection
  $\Gamma[\{S,i,R,W\},\{3,7,9,10\}]$ (with $\det=1>0$, odd bijection permutation
  sign) coexists with UPF $4$-CS having $\det<0$.
  The Jacobian therefore does not have a fixed-sign determinant, and
  \cite[Cor.~5.21]{Vas} does not apply via the fixed-sign route.
  The oscillatory conclusion rests on Recipe~I.

\subsection{Comments on deriving  heuristically     the Boros--Rost Hopf witness $\{W,S_2,i\}$}\label{s:BR}

We have learned at the recent DSABNS 2026 conference that Boros and Rost \cite{BorRost} have proved symbolically the presence of Hopf bifurcations for SIRWS, using a three dimensional mass-action ``Hopf witness" and inheritance.
Since we believe this kind of result provides a major argument in the favor of the unification
of ME and CRN we have been advocating for a while \cite{AAN,VAA,AAHJ,JA,AABH25},
 we have included some comments  about it in this section.
 
 \beR The most relevant UPFs for arriving to the  Boros-Rost Hopf witness of next section do not include the oscillatory core, but the maximal UPF and the third UPF $\{i,W\}$ found in Table \ref{t:UPF}.  The maximal one is where the oscillation might most easily appear, and the third one suggests that the oscillation might still be present after contracting/eliminating R, which turns out to be the case here once the loop is closed by putting S back.
 However, going from the rigorous information provided by BiRNe to potential Hopf witnesses is for now an art, more than a science. \eeR
 
The Boros--Rost Hopf witness $\{S,i,W\}$, which includes the UPF $\{i,W\}$,  is defined by \begin{center}
  \begin{tabular}{ll}
  \textbf{RN} & \textbf{rts} \\
  \hline
  $0 \to W$        & $\Lambda$ \\
  $W \to S2$       & $\gamma\, W$ \\
  $i + W \to i$    & $\nu\, i\, W$ \\
  $S2 + i \to 2i$  & $\beta\, i\, S2$ \\
  $W \to 0$        & $\mu_1 W$ \\
  $S2 \to 0$       & $\mu_2 S2$ \\
  $i \to 0$        & $\mu_3 i$ \\
  \end{tabular}
  \end{center}
Here $S2$ denotes the susceptibles in the contracted model.

  It has  Hopf bifurcation \cite{BorRost}, since
it has only positive \ch\ coefficients at the endemic point, and the Hurwitz determinant/Hopf resultant is zero when
$$\Lambda = \frac{\left(\beta  i+\mu _2\right) \left(\gamma +i \nu +\mu _1\right) \left(\gamma ^2+\gamma  \left(\beta  i+2 i \nu +2 \mu _1+\mu _2\right)+\beta  i \left(i \nu +\mu _1+\mu _3\right)+\left(i \nu +\mu _1\right) \left(i \nu +\mu _1+\mu _2\right)\right)}{\beta  \gamma  i \nu }.$$

Finally, at the endemic point $i$ \sats\ a quadratic equation which has an unique positive root whenever $R_0=\fr{\beta  \gamma  \Lambda}{\mu _2 \mu _3 \left(\gamma +\mu _1\right)} >1,$ and where $R_0$ is precisely the \brn.

\beR The witness is obtained by ``contracting" R in   the loop which causes the SIRWS oscillation $
i \;\longrightarrow\; R \;\longrightarrow\; W \;\longrightarrow\; S \;\longrightarrow\; i
$
   to $
i \;\longrightarrow\; W \;\longrightarrow\; S2 \;\longrightarrow\; i
$, with $S$ having being denoted by $S2$,and
$i\to R\to W$ replaced by $
i \to R \to W \;\leadsto\; i \to W.
$  As noted already, a precise relation to the UPFs, which would allow automatic detection of witnesses, is not available.
\eeR
\begin{figure}[H]
\centering
\renewcommand{\sep}{7.5}
\begin{tikzpicture}[scale=0.980, every node/.style={transform shape},
  node distance=1.8cm,
  species/.style={circle, draw, thick, minimum size=8mm, inner sep=1pt},
  upf/.style={circle, draw, thick, fill=red!15, minimum size=8mm, inner sep=1pt},
  stab/.style={circle, draw, thick, fill=blue!10, minimum size=8mm, inner sep=1pt},
  rxn/.style={-{Stealth[length=2.5mm]}, thick},
  death/.style={-{Stealth[length=2mm]}, thin, dashed, gray},
  proj/.style={-{Stealth[length=3mm]}, very thick, blue!60!black},
  enlarge/.style={-{Stealth[length=3mm]}, very thick, red!60!black},
  lbl/.style={font=\small, midway},
]

\node[stab] (S) at (0, 1.8) {$S$};
\node[upf]  (i) at (1.8, 1.8) {$i$};
\node[upf]  (R) at (1.8, 0) {$R$};
\node[upf]  (W) at (0, 0) {$W$};

\draw[rxn] (S) -- node[lbl, above] {\footnotesize $\beta i S$} (i);
\draw[rxn] (i) -- node[lbl, right] {\footnotesize $\gamma_i i$} (R);
\draw[rxn] (R) -- node[lbl, below] {\footnotesize $\gamma_r R$} (W);
\draw[rxn] (W) -- node[lbl, left]  {\footnotesize $\gamma_w W$} (S);

\draw[death] (S) -- ++(0, 0.7) node[above, font=\tiny] {$\mu_s S$};
\draw[death] (i) -- ++(0.7, 0) node[right, font=\tiny] {$\mu_i i$};
\draw[death] (R) -- ++(0.7, 0) node[right, font=\tiny] {$\mu_r R$};
\draw[death] (W) -- ++(0, -0.7) node[below, font=\tiny] {$\mu_w W$};

\draw[rxn, gray] (-0.9, 1.8) node[left, font=\tiny] {$\Lambda$} -- (S);

\draw[rxn, dotted] (W) to[bend left=25]
  node[lbl, above, font=\tiny, pos=0.5] {$\nu i W$} (R);

\begin{scope}[on background layer]
  \node[draw=red, dashed, rounded corners=6pt,
    fit=(i)(R)(W), inner sep=4pt,
    label={[font=\scriptsize, red!70!black]below left:UPF $\{i,R,W\}$}] {};
\end{scope}

\node[font=\footnotesize\bfseries, anchor=south] at (0.9, 2.7) {$\mathcal{N}_{\mathrm{SIRWS}}$};

\node[stab] (S2) at (\sep, 1.8) {$S\!2$};
\node[upf]  (i2) at (\sep+1.8, 1.8) {$i$};
\node[upf]  (W2) at (\sep+0.9, 0) {$W$};

\draw[rxn] (S2) -- node[lbl, above] {\footnotesize $\beta i S\!2$} (i2);
\draw[rxn] (W2) -- node[lbl, below left, pos=0.4] {\footnotesize $\gamma W$} (S2);
\draw[rxn, dotted] (W2) -- ++(0.9, 0) node[right, font=\tiny] {$\nu i W$};

\draw[death] (S2) -- ++(0, 0.7) node[above, font=\tiny] {$\mu_2 S\!2$};
\draw[death] (i2) -- ++(0.7, 0) node[right, font=\tiny] {$\mu_3 i$};
\draw[death] (W2) -- ++(0, -0.7) node[below, font=\tiny] {$\mu_1 W$};

\draw[rxn, gray] (\sep-0.8, 0) node[left, font=\tiny] {$\Lambda$} -- (W2);

\node[font=\footnotesize\bfseries, anchor=south] at (\sep+0.9, 2.7) {$\mathcal{N}_{\mathrm{BR}}$};

\draw[proj] (2.5, 2.4) -- node[above, font=\small] {$\pi$}
  node[below, font=\scriptsize, text=blue!60!black] {collapse $R$} (\sep-0.7, 2.4);

\draw[enlarge] (\sep-0.7, -0.6) -- node[below, font=\small] {$\mathcal{E}$}
  node[above, font=\scriptsize, text=red!60!black] {inheritance} (2.5, -0.6);

\node[font=\scriptsize, text=gray, anchor=west] at (3.8, 1.0) {
  \begin{tabular}{@{}l@{}}
  $S2 \leftrightarrow S$\\
  $W \leftrightarrow W$\\
  $i \leftrightarrow i$\\
  $i{\to}R{\to}W \rightsquigarrow i{\to}W$
  \end{tabular}
};

\end{tikzpicture}

\caption{Contraction and enlargement between the red dashed SIRWS network on left
$\mathcal{N}_{\mathrm{SIRWS}}$ and the Boros--Rost witness
$\mathcal{N}_{\mathrm{BR}}$ on right.
The projection $\pi$ collapses the immune pathway $i\to R\to W$
into a direct interaction $i\to W$, which is represented however as a consumption,  since $\nu i W$ consumes W,  but
  dotted, to indicate catalysis by $i$.
The enlargement $\mathcal{E}$ (in the sense of Banaji--Boros--Hofbauer)
reverses this, inserting $R$ as an intermediate between $i$ and $W$ and recovering the full
SIRWS dynamics. The Hopf bifurcation transfers
from $\mathcal{N}_{\mathrm{BR}}$ to $\mathcal{N}_{\mathrm{SIRWS}}$ by the inheritance theorems \cite{BBH25}.}
\label{fig:contraction}
\end{figure}


\beO [algorithmic method for automatic detection of Hopf witnesses]
Devise an algorithmic method for automatic detection of Hopf witnesses, starting from  UPFs, or from oscillation cores. \eeO
\begin{remark}[concepts related to contractions in the CRN literature]\label{rem:contraction}
The projection $\pi\colon \mathcal{N}_{\mathrm{SIRWS}} \to \mathcal{N}_{\mathrm{BR}}$
that collapses the pathway $i\to R\to W$ into $i\to W$ is an instance of
\emph{species elimination/contraction} or \emph{intermediate reduction}, which appears
in several related guises in the CRN literature:

\begin{enumerate}[nosep]
\item \textbf{Lumping for model reduction.}
  In chemical kinetics, ``lumping" (
  Li--Rabitz \cite{LR89}) groups species into macro-variables,
  producing a reduced network.
  The identification $\{R,W\}\rightsquigarrow W$ (treating $R$ as a
  transient state towards the waning class) is a form of exact lumping
  when the $R\to W$ transition is fast relative to other timescales.
\item \textbf{Quasi-steady-state and intermediate elimination.}
  Classical QSSA elimination of fast intermediates
  (Briggs--Haldane, Segel--Slemrod \cite{SS89}) contracts pathways
  $A\to B\to C$ to effective reactions $A\to C$ when $B$ equilibrates rapidly.
  Feliu and Wiuf \cite{FW13} formalize this for CRNs,  and emphasize that the  {\bf intermediate elimination/projection} map $\pi$ preserves the
local Jacobian feedback structure while reducing the
network dimension.


\end{enumerate}

\end{remark}

\begin{remark}[Non-uniqueness of Hopf witnesses]
The three--species Boros--Rost admissible reduction provides one explicit mass--action
realization of  the SIRWS model, but is not unique.
In particular, the location of inflow and outflow reactions, as well as the
distribution of intermediate steps along the immune pathway, may be modified
without altering the underlying feedback structure responsible for
instability.

\end{remark}

Finally, the $(W,S2,i)$ Hopf bifurcation turns out to  extend by inheritance to the
SIRWS model \cite{BorRost}.

\section{Other results II: bifurcations in the Capasso-Ruan-Wang family}\label{s:CRW}

We now return to the Capasso-Ruan-Wang family \eqref{Vys}, \eqref{Vyss} of Section \ref{s:CRWmodel}, and review the bifurcation results of \cite{VAA}, obtained with the symbolic Jacobian approach of Section \ref{s:CS} (notebook \texttt{SIRScrw.wl}). The main conclusion is that bifurcations at an endemic point may only occur when the reproduction function \eqr{repF} is larger than one at that point.

\ssec{The symbolic hunt of bifurcations}

 Recall now the core ideas of the method of  \emph{Symbolic Hunt} of bifurcations for positive equilibria \cite{Vas}:
\begin{enumerate}

\item The  Jacobian matrix $Jac$ of \eqref{Ga} is made symbolic,
\begin{equation}Jac:=\Gamma \frac{\partial \mathbf{r}(x)}{\partial x}=G(\mathbf{r}'),
\end{equation}
 with  `free' symbols  $\mathbf{r}'_{jk}$  which replace $\partial r_j / \partial x_k$ (by the monotonicity assumption on rates,  all symbols are strictly positive at a positive equilibrium).

\item The  existence of a positive equilibrium of \eqref{Ga} necessarily requires the existence of a positive right kernel vector \emph{equilibrium flux}  $\bar{\mathbf{r}}>0$ of the stoichiometric matrix $\Gamma$, i.e.,
\begin{equation}\label{eq:FluxVector}
    \Gamma\bar{\mathbf{r}}=0.
\end{equation}

For the \CRW\ model, the flux constraints at an endemic equilibrium $(\bar{s},\bar{i},\bar{r})>0$  are:
\begin{equation}\label{eq:endeqcap}
\begin{cases}
\Lambda=\mu_s\bar{s}+\mu_i\bar{i}+\mu_r \bar{r};\\
f(\bar{s},\bar{i})=h(\bar{i})+\mu_i \bar{i};\\
(\gamma_r+\mu_r) \bar{r}=\gamma_s\bar{s}+h(\bar{i}).
\end{cases}
\end{equation}
Note that \eqref{eq:endeqcap} together with $h(i)>0$ imply that the  existence of an endemic equilibrium  requires
\begin{equation}\label{eq:capassoeqnecessary}\gamma_r+\mu_r>0.
\end{equation}
This condition is satisfied in most models as $\gamma_r=\mu_r=0$ whould imply no mortality and that the individuals will not be leaving the $R$ class.

\item Assuming the existence of an equilibrium flux vector $\bar{\mathbf{r}}$ and of a choice of symbols $\bar{\mathbf{r}}'$ such that $G(\bar{\mathbf{r}}')$ is non-hyperbolic,  for any given parametric rate $r_j(x,p)$ one needs to find $\bar{x}>0$ and parameter values $\bar{p}$ such that
\begin{equation}\label{eq:parameterconstraints}
\begin{cases}
r_j(\bar{x},\bar{p})=\bar{v}_j\\
\partial r_j(\bar{x},\bar{p})/\partial x_k=\bar{v}'_{jk}.
\end{cases}
\end{equation}

The common sense program above culminated in the introduction of ``\emph{parameter-rich kinetics}" in \cite{VasSta}. These are defined  by conditions which ensure,  in the monotone case,  the existence of positive solutions to \eqref{eq:parameterconstraints} for any choice of $\bar{x}>0$, and which include Michaelis--Menten and Hill, but not mass action. For systems that include reaction rates that are in mass action form, we can still use the approach both to exclude bifurcations or as a first guidance to find a bifurcation. The full solution of \eqref{eq:parameterconstraints} must be however pursued case-by-case.

\end{enumerate}

Using this approach, we have proved in \cite{VAA}   that a very common network pattern in ME  implies that the symbolic Jacobian admits purely imaginary eigenvalues.

\beT[symbolic Jacobians of ME models have  purely imaginary eigenvalues, for some parameter set] \label{lem:purelyimaginary}
Consider an epidemiological model which contains at least two populations, called Susceptible $S$ and Infected $I$, and two reactions $1$ and $2$ which imply an ODE
\begin{equation}
\begin{cases}
\dot{s}=-r_1(s,i)+ ...\\
i'=r_1(s,i)-r_2(i)+...\\
...
\end{cases}\end{equation}
where the functions $r_1(s,i)$ and $r_2(i)$ are assumed monotone increasing in $s$, $i$.

Then, there exists a parameter set such that the symbolic Jacobian of such an epidemiological system has purely imaginary eigenvalues.
\eeT

\begin{remark} For systems that are parameter-rich and admit a positive endemic equilibrium, Theorem \ref{lem:purelyimaginary} implies the existence of an equilibrium with purely imaginary eigenvalues. However, ME systems typically include many reactions in mass-action form, therefore Theorem \ref{lem:purelyimaginary} is non-conclusive, and the existence of periodic solutions must be checked case-by-case.
\end{remark}


\ssec{Detailed \RH\  analysis of the symbolic Jacobian of the system \eqref{Vyss}}
The symbolic Jacobian of the system \eqref{Vyss} $
\begin{cases}
s'= \Lambda  -\gamma_s s    - f[s,i] 
+ \gamma_r r -\mu_s s+  i_s i,
\\
i'=
f[s,i] -i(\gamma_i +\mu_i) -T[i], \\
r'= \gamma_s s + i_r   i +T[i] -(\gamma_r +\mu_r) r
\end{cases}$:
\begin{equation}\label{eq:jac}
G=\left(
\begin{array}{ccc}
 -f'_s-\gamma _s-\mu _s & i_s-f'_i & \gamma _r \\
 f'_s & f'_i-\mu _i-h' & 0 \\
 \gamma _s & i_r+T'(i) & -\gamma _r-\mu _r \\
\end{array}
\right)
\end{equation}
reveals that
\begin{remark}\label{obs:gersh}
If $f'_i\le 0$ 
 then the symbolic Jacobian matrix $G$ is weakly-diagonally (column) dominant, i.e. $|G_{ii}|\ge \sum_j |G_{ji}|$, and thus $G$ never possesses eigenvalues with positive real-part.

  For  any bifurcation  to be possible, we must have $f'_i>0$ at an endemic equilibrium.

\end{remark}

 The following detailed Routh-Hurwitz analysis  strengthens Remark \ref{obs:gersh}.

 The characteristic polynomial of the symbolic Jacobian $G$:
$$g(\lambda)=\operatorname{det}(\lambda \operatorname{Id}-G)=c_0\lambda^3+c_1\lambda^2+c_2\lambda+c_3,$$
\begin{equation}
\text{with }
\begin{cases}
\begin{split}
c_0&=1\\
c_1&= f'_s + h' + (\gamma_r +\mu_r+ \gamma_s  + \mu_s+\mu_i )-f'_i \\
c_2&=f'_i(\gamma_r + \gamma_s + \mu_r + \mu_s) - f'_s\gamma_r - f'_sh' - \gamma_rh' - \gamma_sh' - f'_s\mu_i - f'_s\mu_r \\
&- \gamma_r\mu_i - \gamma_s\mu_i - \gamma_r\mu_s - \gamma_s\mu_r - h'\mu_r - h'\mu_s - \mu_i\mu_r - \mu_i \mu_s - \mu_r \mu_s\\
c_3&=\gamma_sh'\mu_r+\gamma_s\mu_i\mu_r+\mu_s  h' \gamma_r +f'_s h' \mu_r +  \mu_i \mu_s \gamma_r+ f'_s \mu_i \mu_r+f'_s \mu_i\gamma_r +\mu_s h' \mu_r+\mu_s \mu_i \mu_r\\&- f'_i (\mu_s \gamma_r+\mu_s\mu_r+\gamma_s\mu_r).
\end{split}
\end{cases}
\end{equation}
We used our Mathematica package to compute these coefficients.

First observation is that both $c_3$ and $c_1$  are destabilized by
$f'_i$. Also,
 rewriting the coefficients $c_2, c_3$ as
\begin{equation}\label{eq:cp_G}
\begin{cases}
c_2=f'_s(h'+\mu_i+\gamma_r+\mu_r)+\mu_s(h'+\gamma_r+\mu_r)-
(f'_i-h'-\mu_i)(\gamma_r+\gamma_s+\mu_r+\mu_s)\\
c_3=f'_s (h' \mu_r +   \mu_i \mu_r+\mu_i\gamma_r)-(f'_i-h'-\mu_i) (\mu_s \gamma_r+\mu_s\mu_r+\gamma_s\mu_r)
\end{cases},
\end{equation}
we refine this to $G_{ii}=f'_i-h'-\mu_i>0$ being  necessary for  zero-eigenvalue bifurcation (also, in the opposite case, note that the positivity of all the coefficients excludes real positive roots). Interestingly, the simplification above reveals that all the coefficients can be written as first order polynomials
in a quantity which is related to $R_0$ and its generalization $R_0(S,i)=\frac{f'_i(S,i)}{\mu_i +h'(i)}$, which are related to the lower triangular block structure of the Jacobian at the DFE.

  A further inspection of \eqref{eq:cp_G} yields
\begin{lemma}\label{lem:zeroeigcap}
The following statements hold:
\begin{enumerate}
    \item For the fully-open model $\mu_s,\mu_i,\mu_r>0$ and the fully-closed model $\mu_s=\mu_i=\mu_r=0$
there is always a choice of $f'_i>0$ such that the symbolic Jacobian $G$ has a change of stability via zero-eigenvalue.
\item Consider the half-open models with $\mu_i>0$, $\mu_s=0$. If  $\gamma_s\mu_r=0$ then there is never a choice of parameters such that the symbolic Jacobian $G$ has a change of stability via zero-eigenvalue.
\end{enumerate}
\end{lemma}
\beR The half-open models in point 2 of Lemma \ref{lem:zeroeigcap} can be thought to model deadly infectious diseases where the death rate of susceptible population is considered negligible.
\eeR

The conclusion of the previous observations is:
\beT [zero-eigenvalue bifurcation requires the reproduction function to be larger than 1 at the endemic point]\label{thm:necessarybif}

A) For a zero-eigenvalue bifurcation to occur at an endemic equilibrium point $(\bar{s},\bar{i},\bar{r})$ of system $\eqref{Vys}$ it is necessary that
\begin{equation}\label{eq:Rcapasso}
G_{ii}(\bar{s},\bar{i})>0 \Eq f'_i(\bar{s},\bar{i})>h'(\bar{i})+\mu_i \Eq R_0(\bar{s},\bar{i})>1.
\end{equation}
Note that either $f$ or $h$ must be nonlinear in $i$.


\eeT
The proof  requires just examining $c_3$ in all the cases mentioned in Lemma \ref{lem:zeroeigcap}.
\beR  Interestingly, the necessary condition \eqr{eq:Rcapasso} for
bifurcation to occur at any endemic equilibrium point is expressed in terms of
the  reproduction function, an object defined using the regular splitting of the transversal block of the Jacobian on the siphon face (Theorem \ref{t:NGM}, Corollary \ref{c:NGM} and Remark \ref{r:repF}).
This happens heuristically due to a ``relay phenomena": the endemic points exist only when the DFE is unstable, which was  observed in many examples, see for example  \cite{AH}, but not proved under general conditions. \eeR

\paragraph{A particular case}
\beT[nonlinearity of T(i) is necessary for bifurcation in the Monod-Haldane class] \label{cor:necessaryhopfMM}
If $f[s,i]$ belongs to the Monod-Haldane class \begin{equation}\label{eq:rationalform}
f[s,i]=\beta s i\frac{1}{1+b i^n}, \quad n\ge 1,
\end{equation}
(which includes in particular Michaelis-Menten for $n=1$),
 then for an equilibrium bifurcation to occur, the nonlinearity of $T[i]$ is necessary.
\eeT

\begin{proof}
Firstly note that
\begin{equation}
f'_i(s,i)=\frac{\beta s(1-b(n-1) i^n)}{(1+b i^n)^2} \le \frac{\beta s}{(1+b i^n)^2}=\frac{\beta s i}{(1+b i^n)^2 i}=\frac{f[s,i]}{(1+b i^n)i}\le\frac{f[s,i]}{i}.
\end{equation}
We used the fact that at an endemic fixed point $i > 0.$ Assume  that $h(i)$ is linear, i.e. $T[i]=0$. Fix an equilibrium $(\bar{s},\bar{i})$. We have
\begin{equation}
G_{ii}(\bar{s},\bar{i})=f'_i(\bar{s},\bar{i})-\gamma_i-\mu_i\le\frac{f(\bar{s},\bar{i})}{\bar{i}}-\gamma_i -\mu_i =\frac{f(\bar{s},\bar{i})-\gamma_i \bar{i}  -\mu_i \bar{i}}{\bar{i}}=0
\end{equation}
and Theorem \ref{thm:necessarybif}
excludes equilibria bifurcations.
\end{proof}

\section{Conclusion}

In this paper, we have advanced the ongoing synthesis between chemical reaction network theory (CRNT) and mathematical epidemiology (ME) by establishing conceptual bridges,
generalizing classical results, and demonstrating the utility of symbolic computational methods for bifurcation analysis in positive ODE systems.

Our first contribution consists of two structural properties of the Jacobian of a positive ODE on any invariant boundary face (equivalently, on the face of a siphon).
Using only the invariance of the face, we proved that the mixed Jacobian block $J_{xy}$ vanishes identically on such faces, yielding a block lower-triangular Jacobian at every point of the face, and, using the positivity of the ODE,
that the transversal diagonal block $J^\perp$ is a Metzler matrix. The first property reduces stability analysis on invariant faces to the separate examination of the transversal and tangential blocks.
The second one is exactly what makes the next generation matrix (NGM) threshold $\rho(FV^{-1})<1$ valid, and we showed by counterexamples that the threshold may fail without it.
Together, they yield a CRNT-flavored generalization of the NGM theorem to arbitrary invariant faces, the automatic existence of regular splittings, a precise sign relation between reproduction functions and transversal stability at every point of the face,
and the Perron--Frobenius fact that transversal stability is lost only through a real eigenvalue, never through a Hopf bifurcation.

Among the other results, our second major theme is the systematic application of the Vassena-Stadler symbolic Jacobian methodology, which decomposes characteristic polynomial coefficients into signed sums
indexed by child selections (matchings between species and reactions). Using this framework, we analyzed two epidemiological models. For the SIRWS model with immunity boosting,
we identified unstable positive feedbacks (UPFs) and their minimal Hurwitz-stable supersets, confirming the structural capacity for Hopf bifurcations.
The recent inheritance results of Boros and Röst, which lift a three-species Hopf witness to the full four-dimensional SIRWS model, illustrate the power of combining child selection
analysis with inheritance theorems. For the broader Capasso–Ruan–Wang SIRS family, we performed a detailed symbolic Routh–Hurwitz analysis of the three-dimensional Jacobian and
derived necessary conditions for bifurcations. In particular, we showed that any endemic bifurcation requires the reproduction function to exceed unity ($R_0 > 1$) and that for the
Monod–Haldane class of incidence functions, nonlinear treatment is necessary for bifurcation. These results demonstrate how symbolic Jacobian methods can yield model-independent
insights.

Throughout, we have emphasized the complementarity between symbolic instability criteria (child selections, UPFs) and inheritance results that transfer bifurcations from reduced
cores to larger networks. The interplay between these approaches, together with the computational tools implemented in our EpidCRN package, provides a powerful workflow for
analyzing positive ODEs arising in epidemiology and beyond.

Several open problems remain. First, since regular splittings of the transversal block always exist, the relevant question becomes the characterization, in CRN terms, of the epidemiologically meaningful ones (Problem 1), which would unify stability analyses across many ME models.
Second, the automatic detection of Hopf witnesses from child selection data (Problem 2) would transform the current heuristic process into a systematic algorithm. Finally, extending
the symbolic Jacobian framework to fully parametric mass-action systems—where the existence of equilibria imposes additional constraints—remains a challenge that inheritance theory
and specialized criteria are beginning to address.

In conclusion, the synergy between CRNT and ME offers not only a richer theoretical understanding of positive ODEs but also practical computational strategies for stability and
bifurcation analysis. The tools and results presented here lay the groundwork for further unification, enabling the transfer of concepts such as siphons, child selections,
and inheritance across disciplines, and opening new avenues for the study of complex dynamical systems in biology and epidemiology.
\section{Declarations}

\textbf{Compliance of Ethical Standard:} The authors have no competing interests to declare that are relevant to the content of this article. The paper is of a theoretical nature so
no experiments were performed and no human or animal subjects were involved.

\appendix

\section{Index of definitions, theorems, examples and open problems}\label{a:index}
{\renewcommand{\glossarysection}[2][]{}%
\printglossary[style=unifiedstyle]}

\section{Background on child selections}\label{a:CS}

\ssec{From critical fragments to child selections: a historical overview}
The idea of explaining oscillations and multistationarity in chemical reaction
networks by isolating small destabilizing substructures has evolved gradually,
with increasing levels of formalization. We review here that evolution,
clarifying terminology and attribution.
\BEN \im {\bf The critical fragments of Ivanova, and Mincheva--Roussel.}
The earliest  formulation of the idea is due to Ivanova, followed by Mincheva and Roussel
\cite{Ivanova,MinRou,MinRoub}.
Working directly at the level of the Jacobian matrix, they introduced the notion
of a \emph{critical fragment}: a minimal subnetwork whose contribution to the
Jacobian is responsible for a loss of stability; the surrounding network may be stabilizing, but cannot compensate for the
instability generated by the fragment.

\im {\bf Jacobian determinant expansions and matchings}
Subsequent work, notably by Mincheva, Banaji, Craciun \cite{MinCra,BanCra}, and collaborators, clarified
the algebraic structure underlying Jacobian minors and determinants.
Using Cauchy--Binet expansions, these works showed that determinants of Jacobian
sub-matrices can be written as sums indexed by matchings between species and
reactions.

\im {\bf Vassena--Stadler: child selections}
Vassena and Stadler introduced \emph{child selections} \cite{Vas,VasSta} as explicit combinatorial
objects indexing the terms appearing in Jacobian determinant expansions.
A child selection is a bijection between a set of species and a set of reactions
in which they appear as reactants.

This formalization represents a conceptual shift: destabilizing mechanisms are
no longer detected indirectly via Jacobian minors, but directly through
well-defined combinatorial structures.
In particular, the sign of the determinant associated with a child selection
can be read off from the stoichiometric submatrix it defines.

\im \cite{BSV}  introduce the notion of an \emph{oscillatory core}, combining:
(i) a destabilizing UPF, and
(ii) a larger child selection whose associated Jacobian block is Hurwitz-stable
and contains the UPF as a principal substructure.

\im {\bf The BiRNe framework} (Golnik, Gatter, Stadler, Vassena) \cite{GSV} builds on the child
selection formalism and provides an algorithmic procedure to enumerate such
structures.

\EEN

\ssec{Proofs of the lemmas of Section \ref{s:VS}}

\begin{proof}[Proof of Lemma \ref{lem:CB-CS-exact}]
By definition of $R$, one has
\[
R_{ji}\neq 0 \iff s_i^j>0.
\]
Therefore, in the Leibniz expansion
\[
\det\!\bigl(R[E_\kappa,\kappa]\bigr)
=
\sum_{J:\kappa\to E_\kappa\ \mathrm{bij.}}
\operatorname{sign}(J)\prod_{i\in\kappa} R_{J(i),i},
\]
the term associated with a bijection $J$ is nonzero if and only if
$R_{J(i),i}\neq 0$ for all $i\in\kappa$, equivalently if and only if
$s_i^{J(i)}>0$ for all $i\in\kappa$. This is exactly the condition that
$J$ be a child-selection bijection.

Now fix such a child selection $(\kappa,E_\kappa,J)$ and let $P_J$ be the corresponding
permutation matrix. Since
\[
\Gamma[\kappa,J]=\Gamma[\kappa,E_\kappa]P_J,
\]
we have
\[
\det(\Gamma[\kappa,J])
=
\operatorname{sign}(J)\det(\Gamma[\kappa,E_\kappa]).
\]
Multiplying by $\prod_{i\in\kappa}R_{J(i),i}$ shows that the contribution of
$(\kappa,E_\kappa,J)$ to the Cauchy--Binet expansion is
\[
\det(\Gamma[\kappa,J])\prod_{i\in\kappa}R_{J(i),i}.
\]
Summing over all $E_\kappa$ and all child-selection bijections $J:\kappa\to E_\kappa$
gives the stated expansion of $\det(G[\kappa])$ as a sum over child selections, with
the permutation signs absorbed into the reordered determinants.
\end{proof}

\begin{proof}[Proof of Lemma \ref{lem:det-sign-instability}]
Let $\Lambda_1,\dots,\Lambda_k$ be the eigenvalues of $A$ with algebraic multiplicity.
Since $A$ is real, nonreal eigenvalues occur in conjugate pairs, each contributing a positive
factor $|\Lambda|^2$ to $\det A$.
If $A$ were Hurwitz stable, then every real eigenvalue would be negative and thus
$\operatorname{sign}(\det A)=(-1)^m$ where $m$ is the number of real eigenvalues; moreover
$k-m$ is even, so $m\equiv k\pmod 2$, hence $(-1)^k\det A>0$. This proves (1) by
contradiction; as $\det A\neq 0$, no eigenvalue is $0$, so non-Hurwitz implies some
$\Re\Lambda>0$.

For (2), suppose $A$ is unstable. If it has a nonreal eigenvalue with $\Re\Lambda>0$, then (i)
holds. Otherwise all eigenvalues with $\Re\Lambda>0$ are real and positive. Let $p$ be the number
of positive real eigenvalues and $q$ the number of negative real eigenvalues (with multiplicity).
Nonreal eigenvalues contribute positive factors to $\det A$, hence $\operatorname{sign}(\det A)=(-1)^q$.
Also $k-(p+q)$ is even, so $k\equiv p+q\pmod 2$, hence $k+q\equiv p\pmod 2$.
Therefore $\sigma(A)=\operatorname{sign}((-1)^k\det A)=(-1)^{k+q}=(-1)^p$.
Under $\sigma(A)=+1$ this forces $p$ even, proving (ii).
\end{proof}

\begin{proof}[Proof of Lemma \ref{lem:UPF-NF-spectrum}]
The parity restrictions follow directly from Lemma~\ref{lem:det-sign-instability}.
Minimality excludes the possibility of multiple unstable real eigenvalues or multiple unstable
pairs, since any such configuration would appear already in a proper principal submatrix,
contradicting the definition of an unstable core.
\end{proof}

\section{Output of \texttt{Epid-CRN} for the examples}\label{a:out}

\ssec{The CRN representation of the Capasso-Ruan-Wang model}
The command \texttt{ODE2RNs[RHS, var, prF]} converts the ODE \eqref{Vyss} to its CRN representation \eqr{Ga}.
The command
\begin{verbatim}
{RN, rts, spe, alp, bet, gam, RNr} = ODE2RNs[RHS, var, prF]
\end{verbatim}
prints the (RN,rts) decomposition fundamental for  CRN methods:
$$\left(
\begin{array}{cc}
 \text{RN} & \text{rts} \\
 \text{i}\to \text{S} & i i_s \\
 0\to \text{S} & \Lambda  \\
 \text{R}\to \text{S} & R \gamma _r \\
 \text{i}\to \text{R} & i i_r \\
 \text{S}\to \text{R} & S \gamma _s \\
 \text{S}\to 0 & S \mu _s \\
 \text{i}\to 0 & i \mu _i \\
 \text{R}\to 0 & R \mu _r \\
 \text{i}+\text{S}\to 2 \text{i} & f(S,i) \\
 \text{i}\to \text{R} & T(i) \\
\end{array}
\right)
$$

It also outputs the spe(cies), the strings form of var, necessary in parsing reactions, alp, bet  which are the matrices of stoichiometric coefficients corresponding to the LHS and RHS of the reactions, the stoichiometric matrix gam=bet-alp, and an alternative form RNr of the reactions, in which terms common to the LHS and RHS (called catalyzers) are written above the arrow, instead of being added up on both sides, as in the first, mass-action inspired representation.

\ssec{Oscillatory cores of SIRWS}
The \texttt{oscRI} script (Recipe~I oscillatory cores) applied to the SIRWS model of Section \ref{s:SIRWS} returns (see Section \ref{s:VSsirws}):
\begin{verbatim}
Recipe I holds: 1 oscillatory core of Class I
RI#1 stable 2-CS: species={S,i}  reactions={be*i*S, mus*S}
     det=1  Sk= -1  -1
             1   0
  contains UPF 1-CS: species={i}  reactions={be*i*S}  det=1  Sk= 1
\end{verbatim}

The scripts \texttt{oscRII[nCS, gH, spe, rts]} and \texttt{fixSgn[nCS, n]} return:
\begin{verbatim}
Recipe II: no cores of Class II found
fixed sign by order= {{1,True,{1}},{2,False,{-1,1}},
                      {3,False,{-1,1}},{4,False,{-1,1}}}
\end{verbatim}

{\small
 \bibliographystyle{amsalpha}
 \bibliography{ref}
}
\AtEndDocument{%
  \glsaddall  
  }
\end{document}